\documentclass[12pt]{amsart}

\usepackage[hmargin=0.8in,height=8.6in]{geometry}
\usepackage{amssymb,amsthm, times}
\usepackage{delarray,verbatim}
\usepackage{srcltx}
\usepackage{ifpdf}
\ifpdf
\usepackage[pdftex]{graphicx}
 \else
\usepackage[dvips]{graphicx}
 \fi

\usepackage{ifpdf}
\usepackage{color}
\definecolor{webgreen}{rgb}{0,.5,0}
\definecolor{webbrown}{rgb}{.8,0,0}
\definecolor{emphcolor}{rgb}{0.95,0.95,0.95}

\usepackage{hyperref}
\hypersetup{%
          colorlinks=true,
          linkcolor=webbrown,
          filecolor=webbrown,
          citecolor=webgreen,
          breaklinks=true}
\ifpdf \hypersetup{pdftex,
            pdfstartview=FitH, 
            bookmarksopen=true,
            bookmarksnumbered=true
} \else \hypersetup{dvips} \fi

\newcommand {\ME}{\mathbb{E}^{x}}

\renewcommand{\S}{\mathcal{S}}

\numberwithin{equation}{section}

\newtheorem{proposition}{Proposition}[section]

\newtheorem{remark}{Remark}[section]
\newtheorem{lemma}{Lemma}[section]

\numberwithin{remark}{section} \numberwithin{proposition}{section}
\numberwithin{corollary}{section}
\newcommand {\R}{\mathbb{R}}

\newcommand {\F}{\mathcal{F}}

\newcommand {\N}{\mathbb{N}}
\newcommand {\p}{\mathbb{P}}

\newcommand {\E}{\mathbb{E}}

\newcommand{\diff}{{\rm d}}

\newcommand{\lev}{L\'{e}vy }

\title[Precautionary Measures]{Precautionary Measures for Credit Risk Management in Jump Models}
\author[M. Egami]{Masahiko Egami}
\address[M. Egami]{Graduate School of Economics,
Kyoto University, Sakyo-Ku, Kyoto, 606-8501, Japan }
\email{egami@econ.kyoto-u.ac.jp}
\thanks{
M.\ Egami is in part supported
by Grant-in-Aid for Scientific Research (B) No.\ 22330098 and (C)
No.\ 20530340, Japan Society for the Promotion of Science.  K.
Yamazaki is in part supported by Grant-in-Aid for Young Scientists
(B) No.\ 22710143, the Ministry of Education, Culture, Sports,
Science and Technology, and by Grant-in-Aid for Scientific Research (B) No.\  2271014, Japan Society for the Promotion of Science.  The authors thank Ning Cai, Masaaki Kijima, Michael Ludkovski, Goran Peskir and the anonymous referee  for helpful suggestions and remarks.}
\author[K. Yamazaki]{Kazutoshi Yamazaki }
\address[K. Yamazaki]{Center for the Study of Finance and Insurance,
Osaka University, 1-3 Machikaneyama-cho, Toyonaka City, Osaka 560-8531, Japan}
\email{k-yamazaki@sigmath.es.osaka-u.ac.jp}
\date{\today}

\begin{document}

\begin{abstract}
\noindent Sustaining efficiency and stability by properly controlling the equity
to asset ratio is one of the most important and difficult challenges
in bank management.  Due to unexpected and abrupt decline of asset
values, a bank must closely monitor its net worth as well as market
conditions, and one of its important concerns is when to raise more
capital so as not to violate capital adequacy requirements.  In this
paper, we model the tradeoff between avoiding costs of delay and
premature capital raising, and solve the corresponding optimal
stopping problem.  In order to model defaults in a bank's loan/credit
business portfolios, we represent its net worth by \lev
processes, and solve explicitly for the double exponential jump
diffusion process and for a general spectrally negative \lev process.

\end{abstract}

\maketitle \noindent \small{\textbf{Key words:} Credit risk
management; Double exponential
jump diffusion; Spectrally negative \lev processes; Scale functions;
Optimal stopping\\
\noindent Mathematics Subject Classification (2000) : Primary: 60G40
Secondary: 60J75 }\\
%
%

\section{Introduction} \label{sec:introduction}

As an aftermath of the recent devastating financial crisis,  more
sophisticated risk management practices are now being required under
the Basel II accord.  In order to satisfy the capital adequacy
requirements, a bank needs to closely monitor how much of its
asset values has been damaged; it needs to examine whether
it maintains sufficient equity values or needs to start enhancing
its equity to asset ratio by raising more capital
and/or selling its assets.  Due to unexpected sharp declines in
asset values as experienced in the fall of 2008, optimally
determining when to undertake the action is an important and
difficult problem. In this paper, we give a new framework for this
problem and obtain its solutions explicitly.

We propose an alarm system that determines when a bank needs to start enhancing its own
capital ratio. We use \lev processes with jumps in order to
model defaults in its loan/credit assets and sharp declines in
their values under unstable market conditions. Because of their
negative jumps and the necessity to allow time for completing its
capital reinforcement plans, early practical action is needed to reduce the risk of violating
 the capital adequacy requirements. On
the other hand, there is also a cost of premature undertaking. If
the action is taken too quickly, it may run a risk of
incurring a large amount of opportunity costs including burgeoning
administrative and monitoring expenses. In other words, we need to
solve this tradeoff in order to implement this alarm system.

In this paper, we properly quantify the costs of delay and premature undertaking and set a well-defined objective function
that models this tradeoff. Our
problem is to obtain a stopping time that minimizes the
objective function.  We expect that this precautionary measure gives
a new framework in risk management.

\subsection{Problem}
Let $(\Omega, \F, \p)$ be a complete probability space on which a
L\'{e}vy process $X=\{X_t; t\ge 0\}$ is defined.  We represent, by $X$, a bank's \emph{net worth} or \emph{equity
capital} allocated to its loan/credit business and model the defaults in its credit portfolio in terms of the negative jumps.  For example, for a given standard Brownian motion $B=\{B_t; t\ge 0\}$ and a jump process $J=\{J_t; t\ge 0\}$ independent of $B$, if it admits a decomposition
\begin{equation*}
  X_t=x+ \mu t+\sigma B_t + J_t, \quad 0\le t <\infty \quad \textrm{and} \quad X_0=x
\end{equation*}
for some $\mu \in \R$ and $\sigma \geq 0$, then $J_t$ models the defaults as well as rapid increase in capital whereas the
non-jump terms $\mu t$ and $\sigma B_t$ represent, respectively, the
growth of the capital (through the cash flows from its credit
portfolio) and its fluctuations caused by non-default events (e.g.,
change in interest rates). 

Since $X$ is spatially homogeneous, we may assume, without loss of
generality, that the first time $X$ \emph{reaches or goes below zero}
signifies the event that the net capital requirement is
violated. We call this the \emph{violation event} and denote it by
\[
\theta:=\inf\{t\ge 0: X_t \leq 0\}
\]
where we assume $\inf \varnothing = \infty$.
Let $\mathbb{F}=(\F_t)_{t\ge 0}$ be the filtration generated by $X$.  Then $\theta$ is an $\mathbb{F}$-stopping time taking values on $[0,\infty]$. We denote by $\S$ the set of all stopping times \emph{smaller than or equal to} the violation event; namely,
\begin{align*}
\S : = \left\{ \textrm{stopping time }\tau : \tau \leq \theta \; a.s.\right\}.
\end{align*}

We only need to consider stopping times in $\S$ because the violation event is observable and the game is over once it happens.
By taking advantage of this, we see that the problem can be reduced to a well-defined optimal stopping problem; see Section \ref{sec:model}.  Our goal is to obtain among $\S$ the \emph{alarm time} that minimizes the two costs we describe below.

The first cost we consider is the risk that the alarm will be triggered
\emph{at or after} the violation event:
\begin{align*}
R^{(q)}_x(\tau) :=  \E^x \left[ e^{-q\theta} 1_{\left\{ \tau \geq \theta, \, \theta < \infty \right\}}\right], \quad \tau \in \S.
\end{align*}
Here  $q\in [0, \infty)$ is a discount rate and $\p^x$ is the
probability measure and $\E^x$ is the expectation under which the process starts at $X_0=x$.  We
call this the \emph{violation risk}. In particular, when $q=0$, it can be reduced  under a suitable condition to the probability of
the event $\{\tau \geq \theta\}$; see Section \ref{sec:model}.

The second cost relates to premature undertaking measured by
\begin{equation*}
\quad  H_x^{(q,
h)}(\tau):=\E^x\left[1_{\{\tau<\infty\}}\int_\tau^\theta
e^{-qt}h(X_t)\diff t \right], \quad \tau \in \S.
\end{equation*}
We shall call this the \emph{regret}, and here we assume $h:(0,\infty)\rightarrow [0,\infty)$ to be \emph{continuous}, \emph{non-decreasing} and 
\begin{align}
\E^x \left[ \int_0^\theta e^{-qt} h(X_t) \diff t \right] < \infty, \quad x > 0.  \label{cond_integrability}
\end{align}
The monotonicity
assumption reflects the fact that, if a bank has a higher capital
value $X$, then it naturally has better access to high quality
assets and hence the opportunity cost $h(\cdot)$ becomes higher
accordingly. In particular, when $h \equiv 1$ (i.e., $h(x) = 1$ for every $x > 0$), we have
\begin{align}
H^{(0,1)}_x (\tau) = \E^x \theta  - \E^x \tau \quad \textrm{and} \quad H^{(q,1)}_x (\tau) = \frac 1 q \left( \E^x [ e^{-q \tau} ] - \E^x [
e^{-q \theta} ] \right), \quad q > 0 \label{regret_h_1}
\end{align}
where the former is well-defined by \eqref{cond_integrability}.

Now, using some fixed weight $\gamma >0$, we consider a linear
combination of these two costs described above:
\begin{align}
U_x^{(q,h)}({\tau,\gamma}) := R^{(q)}_x(\tau) + \gamma H_x^{(q,h)} (\tau), \quad \tau \in \S.
\label{def_u}
\end{align}We solve the problem of minimizing \eqref{def_u} for the \emph{double exponential jump
diffusion process} and a general \emph{spectrally negative \lev process}.  The objective function is finite thanks to the integrability assumption \eqref{cond_integrability}, and hence the problem is well-defined.

The form of this objective function in \eqref{def_u} has an origin from the
\emph{Bayes risk} in mathematical statistics.  In the Bayesian formulation of \emph{change-point detection}, the Bayes risk is defined as a linear combination of the expected detection delay and  false alarm probability.  In \emph{sequential
hypothesis testing}, it is a linear combination of the expected sample size and misdiagnosis probability.  The optimal
solutions in these problems are those stopping times that minimize
the corresponding Bayes risks. Namely, the tradeoff between promptness and accuracy is modeled in terms of the
Bayes risk. Similarly, in our problem, we model the tradeoff
between the violation risk and regret by their linear
combination $U$.

We first consider the double exponential jump
diffusion process, a \lev process consisting of a Brownian motion and a compound
Poisson process with positive and negative exponentially-distributed
jumps.  We consider this classical model as an excellent starting point mainly due to a number of existing analytical properties and the fact that the results can potentially be extended to the hyper-exponential jump diffusion model (HEM) and more generally to the phase-type \lev model.
Due to the memoryless property of the exponential
distribution, the distributions of the first passage times and
overshoots by this process can be obtained explicitly (Kou and Wang
\cite{Kou_Wang_2003}).  It is this property that leads to analytical
solutions in various problems that would not be possible for other
jump processes.  Kou and Wang \cite{Kou_Wang_2004} used this process
as an underlying asset and obtained a closed form solution to the
perpetual American option and the Laplace transforms of lookback and
barrier options.  Sepp \cite{Sepp_2004} derived explicit pricing formulas for double-barrier and
double-touch options with time-dependent rebates. See also Lipton and Sepp \cite{Lipton_Sepp_2009} for applications of its multi-dimensional version in credit risk.  Some of the results for the double-exponential jump diffusion process have been extended to the HEM and phase-type models, for example, by Cai et al.\ \cite{Cai_2009,Cai_2009_2} and Asmussen et al.\ \cite{Asmussen_2004}. 

%
%

We then consider a general spectrally negative \lev process, or a \lev process with only negative jumps.  Because we are interested in defaults,
the restriction to negative jumps does not lose much reality in
modeling. We also see that positive jumps do
not have much influence on the solutions. We shall utilize the
\emph{scale function} to simplify the problem and obtain analytical solutions.   In order to identify the candidate optimal strategy, we shall apply \emph{continuous and smooth fit} for the cases $X$ has paths of bounded and unbounded variation, respectively.  The scale function is an important tool in most spectrally negative \lev models and can be calculated via algorithms such as Surya \cite{Surya_2008} and Egami and Yamazaki \cite{Egami_Yamazaki_2010_2}.

\subsection{Literature review}

Our model is original, and, to the best of our knowledge, the
objective function defined in (\ref{def_u}) cannot be found
elsewhere. It is, however, relevant to the problem, arising in the
optimal capital structure framework, of determining the endogenous
bankruptcy levels. The original diffusion model was first proposed
by Leland \cite{Leland_1994} and Leland and Toft
\cite{Leland_Toft_1996}, and it was extended, via the Wiener-Hopf
factorization, to the model with jumps by Hilberink and
Rogers \cite{Hilberink_Rogers_2002}.
 Kyprianou and Surya
\cite{Kyprianou_Surya_2007} studied the case with a general spectrally
negative \lev process. In their problems, the continuous and smooth fit principle is a main
tool in obtaining the optimal bankruptcy levels. Chen and Kou
\cite{Chen_Kou_2009} and Dao and Jeanblanc \cite{Dao_Jeanblanc_2006}, in particular,
focus on the double exponential jump diffusion case.

In the \emph{insurance} literature, as exemplified by the
\emph{Cramer-Lundberg} model, the compound Poisson process is
commonly used to model the surplus of an insurance firm. Recently,
more general forms of jump processes are also used (e.g., Huzak et
al.\ \cite{huzak2004} and Jang \cite{jang2007}).  For generalizations to the spectrally negative \lev model, see Avram et al.\ \cite{Avram_et_al_2007}, Kyprianou and Palmowski \cite{Kyprianou_2007}, and Loeffen \cite{Loeffen_2008}. The literature
also includes computations of ruin probabilities and extensions to
jumps with heavy-tailed distributions; see Embrechts et al.\
\cite{embrechts-book} and references therein.  See also Schmidli
\cite{schmid-book} for a survey on stochastic control problems in
insurance.

Mathematical statistics problems as exemplified by sequential
testing and change-point detection have a long history. It dates
back to 1948 when Wald and Wolfowitz \cite{Wald_Wolfowitz_1948,
Wald_Wolfowitz_1950} used the Bayesian approach and
proved the optimality of the sequential probability ratio test
(SPRT).  There are essentially two problems, the
\emph{Bayesian} and the \emph{variational} (or the
\emph{fixed-error}) problems; the former minimizes the Bayes risk
while the latter minimizes the expected detection delay (or the sample
size) subject to a constraint that the error probability is
smaller than some given threshold.  For comprehensive surveys and references, we refer
the reader to Peskir and Shiryaev \cite{Peskir_Shiryaev_2006} and
Shiryaev \cite{Shiryaev_2008}.  Our problem was originally
motivated by the Bayesian problem. However, it is also possible to consider its variational version where
the regret needs to be minimized on constraint that the violation
risk is bounded by some threshold.

Optimal stopping problems involving jumps (including the discrete-time model) are, in general, analytically intractable owing to the difficulty in obtaining the overshoot distribution.  This is true in our problem and in the literatures introduced above.  For example, in sequential testing and change-point detection, explicit solutions can be realized only in the Wiener case.  For this reason, recent research focuses on obtaining asymptotically optimal solutions by utilizing renewal theory; see, for example, Baron and Tartakovsky \cite{Baron_Tartakovsky_2006}, Baum and Veeravalli \cite{Baum_Veeravalli_1994}, Lai \cite{Lai_1975} and Yamazaki \cite{Yamazaki_2009}.  Although we do not address in this paper, asymptotically optimal solutions to our problem may be pursued for a more general class of \lev processes via renewal theory.  We refer the reader to Gut \cite{Gut_2009} for the overshoot distribution of random walks and Siegmund \cite{Siegmund_1985} and Woodroofe \cite{Woodroofe_1982} for more general cases in nonlinear renewal theory.

\subsection{Outline}
The rest of the paper is structured as in the following. We first
give an optimal stopping model for a general \lev process in the next section. Section \ref{sec:double_exponential} focuses on the double exponential jump
diffusion process and solve for the case when $h\equiv 1$. Section \ref{sec:spectrally_negative} considers the case when the process is a general spectrally negative \lev process; we obtain the solution explicitly in terms of the scale function for a general $h$.
We conclude with numerical results in Section
\ref{sec:numerical_results}.  Long proofs are deferred to the
appendix.

\section{Mathematical Model} \label{sec:model}

In this section, we first reduce our problem to an optimal stopping problem and illustrate the continuous and smooth fit approach to solve it.

\subsection{Reduction to optimal stopping}
Fix $\tau \in \S$, $q > 0$ and $x > 0$. The violation risk is
\begin{align*}
R^{(q)}_x(\tau) =  \E^x \left[ e^{-q \theta} 1_{\left\{ \tau \geq \theta, \, \theta < \infty \right\}}\right]  = \E^x \left[ e^{-q \theta} 1_{\left\{ \tau = \theta, \, \theta < \infty \right\}}\right] = \E^x \left[ e^{-q \tau} 1_{\left\{ \tau = \theta, \, \tau < \infty \right\}}\right]
\end{align*}
where the third equality follows because $\tau \leq \theta$ a.s.\ by definition. Moreover, we have
\begin{align*}
1_{\left\{ \tau = \theta, \, \tau < \infty \right\}} = 1_{\{ X_\tau \leq 0, \, \tau < \infty \}} \quad a.s.
\end{align*}
because
\begin{align*}
\left\{ \tau = \theta, \, \tau < \infty \right\} = \left\{ X_\tau \leq 0, \, \tau = \theta, \, \tau < \infty \right\}
\end{align*}
and
\begin{align*}
\left\{ X_\tau \leq 0, \, \tau < \infty \right\} = \left\{ X_\tau \leq 0, \, \tau \leq \theta, \, \tau < \infty \right\} = \left\{ X_\tau \leq 0, \, \tau = \theta, \, \tau < \infty \right\}
\end{align*}
where the first equality holds because $\tau \in \S$ and the second equality holds by the definition of $\theta$. Hence, we have
\begin{align}
R^{(q)}_x(\tau) = \E^x \left[ e^{-q \tau} 1_{\left\{ X_\tau \leq 0, \, \tau < \infty \right\}}\right], \quad \tau \in \S. \label{def_indicator}
\end{align}

For the regret, by the strong Markov property of
$X$ at time $\tau$, we have
\begin{align}
\quad  H_x^{(q,
h)}(\tau):=\E^x\left[1_{\left\{ \tau < \infty \right\}} \int_\tau^{\theta}
e^{-qt}h(X_t)\diff t \right]
 = \E^x \left[ e^{-q \tau} Q^{(q,h)}(X_\tau) 1_{\left\{ \tau < \infty \right\}} \right] \label{def_Q}
\end{align}
where $\left\{ Q^{(q,h)}(X_t); t \geq 0 \right\}$ is an $\mathbb{F}$-adapted Markov process such that
\begin{align*}
Q^{(q,h)}(x) := \E^x\left[ \int_0^{\theta} e^{-qt} h(X_t) \diff t \right], \quad x \in \R.
\end{align*}
 Therefore, by (\ref{def_indicator})-(\ref{def_Q}), if we let
\begin{align}
G(x) := 1_{\{ x \leq 0 \}} + \gamma Q^{(q,h)}(x) 1_{\{ x > 0 \}}, \quad x \in \R \label{def_stopping_value}
\end{align}
denote the cost of stopping, we can rewrite the objective function \eqref{def_u} as
\[
U_x^{(q,h)}({\tau,\gamma})
= \E^x \left[ e^{-q \tau} G(X_\tau) 1_{\left\{ \tau < \infty \right\}} \right], \quad \tau \in \S.
\]
Our problem is to obtain
\begin{align*}
\inf_{\tau \in \S} \E^x \left[ e^{-q \tau} G(X_\tau) 1_{\left\{ \tau < \infty \right\}} \right]
\end{align*}
and an optimal stopping time $\tau^* \in \S$ that attains it if such a stopping time exists.
It is easy to see that $G(x)$ is non-decreasing on $(0,\infty)$ because $h(x)$ is.  If $G(0+) \geq 1$, then clearly $\theta$ is optimal.  Therefore, we ignore the trivial case and assume throughout this paper that
\begin{align}
G(0+) < 1. \label{assumption_G}
\end{align}
As we will see later,  when $X$ has paths of unbounded variation, $G(0+)=0$ and the assumption above is automatically satisfied.

The problem can be naturally extended to the undiscounted case with
$q=0$.  The integrability assumption \eqref{cond_integrability} implies $\E^x \theta < \infty$ (without this assumption $G(x) = \infty$ and the problem becomes trivial). This also implies $\theta < \infty$ a.s.\ and the violation risk reduces to the probability
\begin{align*}
R^{(0)}_x(\tau) = \p^x \left\{ \tau \geq \theta \right\} = \p^x \left\{ \tau = \theta \right\} = \p^x \left\{ X_\tau \leq 0 \right\}, \quad \tau \in \S.
\end{align*}
We shall study the case $q=0$ for the double exponential jump diffusion process in Section \ref{sec:double_exponential}.

\subsection{Obtaining optimal strategy via continuous and smooth fit}
Similarly to obtaining the optimal bankruptcy levels in Leland \cite{Leland_1994}, Leland and Toft \cite{Leland_Toft_1996}, Hilberink and Rogers \cite{Hilberink_Rogers_2002} and Kyprianou and Surya
\cite{Kyprianou_Surya_2007}, the continuous and smooth fit principle will be a useful tool in our problem.
Focusing on the set of \emph{threshold strategies} defined by the first time the process reaches or goes below some fixed threshold, say $A$,
\[
\tau_A:=\inf\{t\ge 0: X_t \leq A\}, \quad A \geq 0,
\]
 we choose the optimal threshold level that satisfies the continuous or smooth fit condition and then verify the optimality of the corresponding strategy.

Let the expected value corresponding to the threshold strategy $\tau_{A}$ for fixed $A > 0$ be
\begin{align*}
\phi_A(x) := U_x^{(q,h)} (\tau_{A}, \gamma), \quad x \in \R,
\end{align*}
and the difference between the continuation and stopping values be
\begin{align} \label{diff_v_g}
\delta_A(x) := \phi_A(x)- G(x) = R_x^{(q)}(\tau_{A})  -
\gamma \E^x \left[ \int^{\tau_A}_0 e^{-qt} h(X_t) \diff t \right], \quad 0 < A < x.
\end{align}
We then have
\begin{align}
\phi_A (x) = \left\{ \begin{array}{ll} G(x) + \delta_A(x), & x > A, \\
G(x), & 0 < x \leq A, \\
1, & x \leq 0. \end{array} \right.  \label{phi_A}
\end{align}
The continuous and smooth fit conditions are $\delta_A(A+) = 0$ and $\delta_A'(A+) = 0$, respectively.
For a comprehensive account of continuous and smooth fit principle, see Peskir and Shiryaev \cite{Peskir_Shiryaev_2006,Peskir_Shiryaev_2000, Peskir_Shiryaev_2002}.

\subsection{Extension to the geometric model}  It should be noted that a version of this problem with an \emph{exponential \lev process} $Y = \left\{Y_t = \exp(X_t); t \geq 0 \right\}$ and a slightly modified violation time
\begin{align*}
\widetilde{\theta} = \inf \left\{ t \geq 0: Y_t \leq a\right\}
\end{align*}
for some $a > 0$ can be modeled in the same framework.  Indeed, defining a shifted \lev process 
\begin{align*}
\widetilde{X} := \left\{\widetilde{X}_t = X_t - \log a; t \geq 0 \right\},
\end{align*}
we have
\begin{align*}
\widetilde{\theta} = \inf \left\{ t \geq 0:\widetilde{X}_t  \leq 0 \right\}.
\end{align*}
Moreover, the regret function can be expressed in terms of $\widetilde{X}$ by replacing $h(x)$ with $\widetilde{h}(x) = h(\exp (x + \log a))$ for every $x > 0$. The continuity and non-decreasing properties remain valid because of the property of the exponential function.

\section{Double Exponential Jump Diffusion}  \label{sec:double_exponential}
In this section, we consider the double exponential jump diffusion
model that features exponential-type jumps in both positive and
negative directions.  We first summarize the results from Kou and
Wang \cite{Kou_Wang_2003} and obtain explicit representations of
our violation risk and regret.  We then find analytically the
optimal strategy both when $q > 0$ and when $q = 0$.  We assume
throughout this section that $h\equiv 1$, i.e., the regret function
reduces to (\ref{regret_h_1}).

\subsection{Double exponential jump diffusion}
The double exponential jump diffusion process is a \lev process of the form
\begin{align}
X_t := x  + \mu t + \sigma B_t + \sum_{i=1}^{N_t} Z_i, \quad t \geq 0 \label{double_exponential_x}
\end{align}
where $\mu \in \R$, $\sigma > 0$, $B = \left\{ B_t; t \geq 0 \right\}$ is a standard Brownian motion, $N = \left\{ N_t; t \geq 0 \right\}$
is a Poisson process with parameter $\lambda >
0$ and $Z = \left\{ Z_i; i \in \mathbb{N} \right\}$ is a sequence
of i.i.d.\ random variables having a \emph{double exponential distribution}
with its density
\begin{align}
f (z)  := p \eta_- e^{\eta_- z} 1_{\left\{z < 0 \right\}} + (1-p)
\eta_+ e^{- \eta_+ z} 1_{\left\{z > 0 \right\}}, \quad z \in \R, \label{density_f}
\end{align}
for some $\eta_-, \eta_+ \geq 0$ and $p \in [0,1]$. Here $B$, $N$ and $Z$ are assumed to be mutually independent.

The \emph{Laplace exponent} of this process is given by
\begin{align}\label{laplace_exponent_double_exponential}
\psi(\beta) := \E^0 \left[ e^{\beta X_1}\right] = \mu \beta + \frac 1 2 \sigma^2 \beta^2  + \lambda
\left( \frac {p \eta_-} {\eta_- + \beta} + \frac {(1-p) \eta_+}
{\eta_+ - \beta} -1\right), \quad \beta \in \R.
\end{align}
We later see that the Laplace exponent and its
inverse function are useful tools in simplifying the problem and
characterizing the structure of the optimal solution.

Fix $q > 0$. There are four roots of $\psi ( \beta ) = q$, and in particular we focus on $\xi_{1,q}$ and $\xi_{2,q}$ such that
\begin{align*}
0 < \xi_{1,q} < \eta_- < \xi_{2,q} < \infty \quad \textrm{and} \quad
\psi(-\xi_{1,q}) = \psi(-\xi_{2,q}) = q.
\end{align*}
Suppose that
the \emph{overall drift} is denoted by $\overline{u} := \E^0 [X_1]$, then it becomes
\begin{align*}
\overline{u} = \mu + \lambda \left( - \frac p {\eta_-} + \frac {1-p}
{\eta_+}\right),
\end{align*}
and
\begin{align}\label{limit_psi}
\xi_{1,q} \rightarrow \xi_{1,0} \left\{ \begin{array}{ll} = 0, & \overline{u}
\leq 0 \\ > 0, &\overline{u} > 0 \end{array} \right\} \quad
\textrm{and} \quad \xi_{2,q} \rightarrow \xi_{2,0} \quad \textrm{as} \; q \rightarrow 0
\end{align}
for some $\xi_{1,0}$ and $\xi_{2,0}$ satisfying
\begin{align*}
 0 \leq \xi_{1,0} < \eta_- < \xi_{2,0} < \infty \quad \textrm{and} \quad  \psi(-\xi_{1,0}) = \psi(-\xi_{2,0}) = 0;
\end{align*}
see Figure \ref{fig:xi} for an illustration. When $\overline{u} < 0$, by (\ref{limit_psi}), l'H\^{o}pital's rule and $\psi(-\xi_{1, q})=q$,
we have
\begin{align}
\frac {\xi_{1,q}} {q} \xrightarrow{q \downarrow 0} - \frac 1 {\psi'(0)} = - \frac 1 {\overline{u}} = \frac 1 {|\overline{u}|}. \label{convergence_xi}
\end{align}
We will see that these roots characterize the optimal strategies; the optimal threshold levels can
be expressed in terms of $\xi_{1,q}$ and $\xi_{2,q}$ when $q > 0$
and $\xi_{2,0}$ and $\overline{u}$ when $q=0$.

\begin{figure}[htbp]
\begin{center}
\begin{minipage}{1.0\textwidth}
\centering
\begin{tabular}{cc}
\includegraphics[scale=0.5]{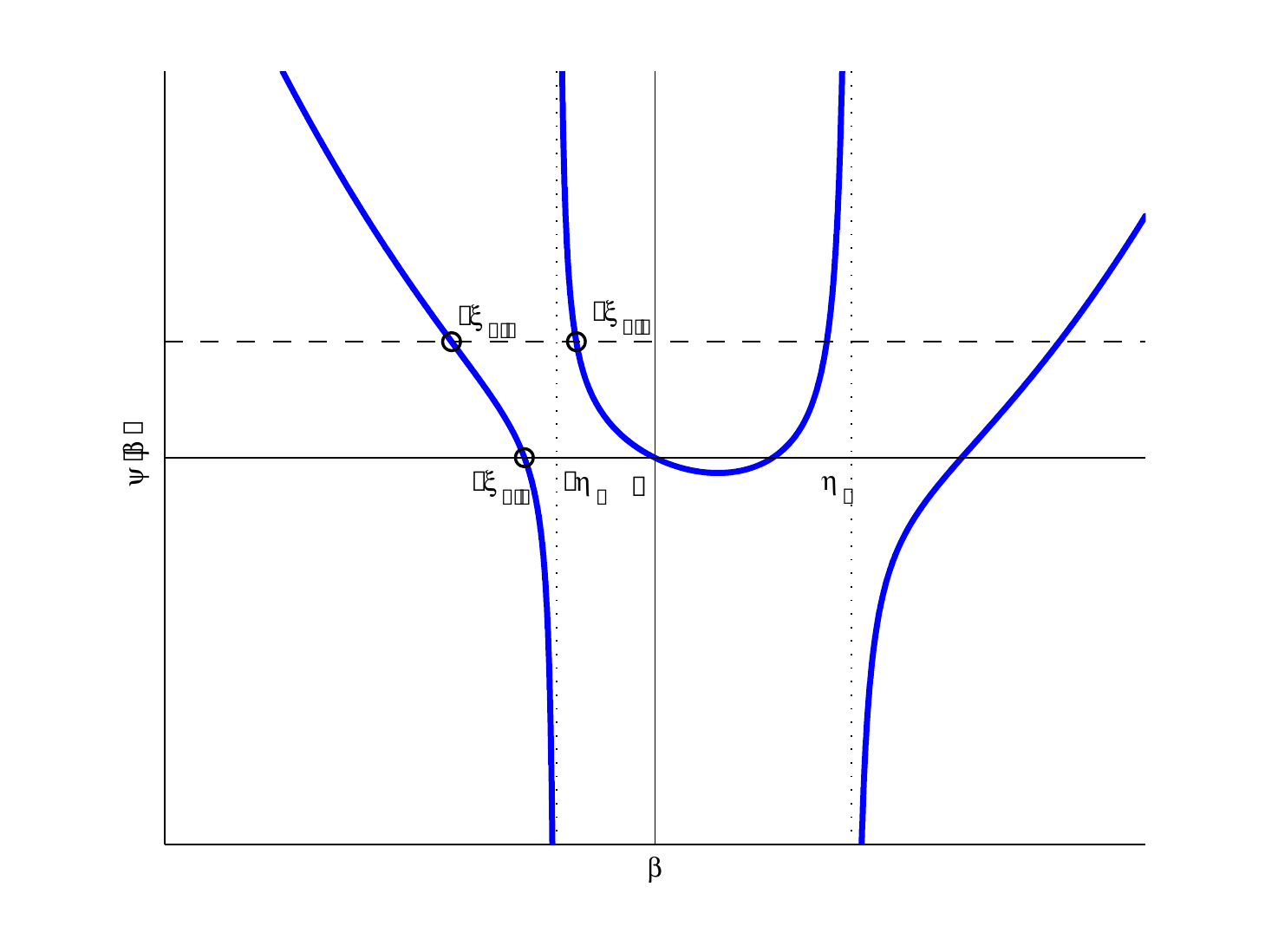}  & \includegraphics[scale=0.5]{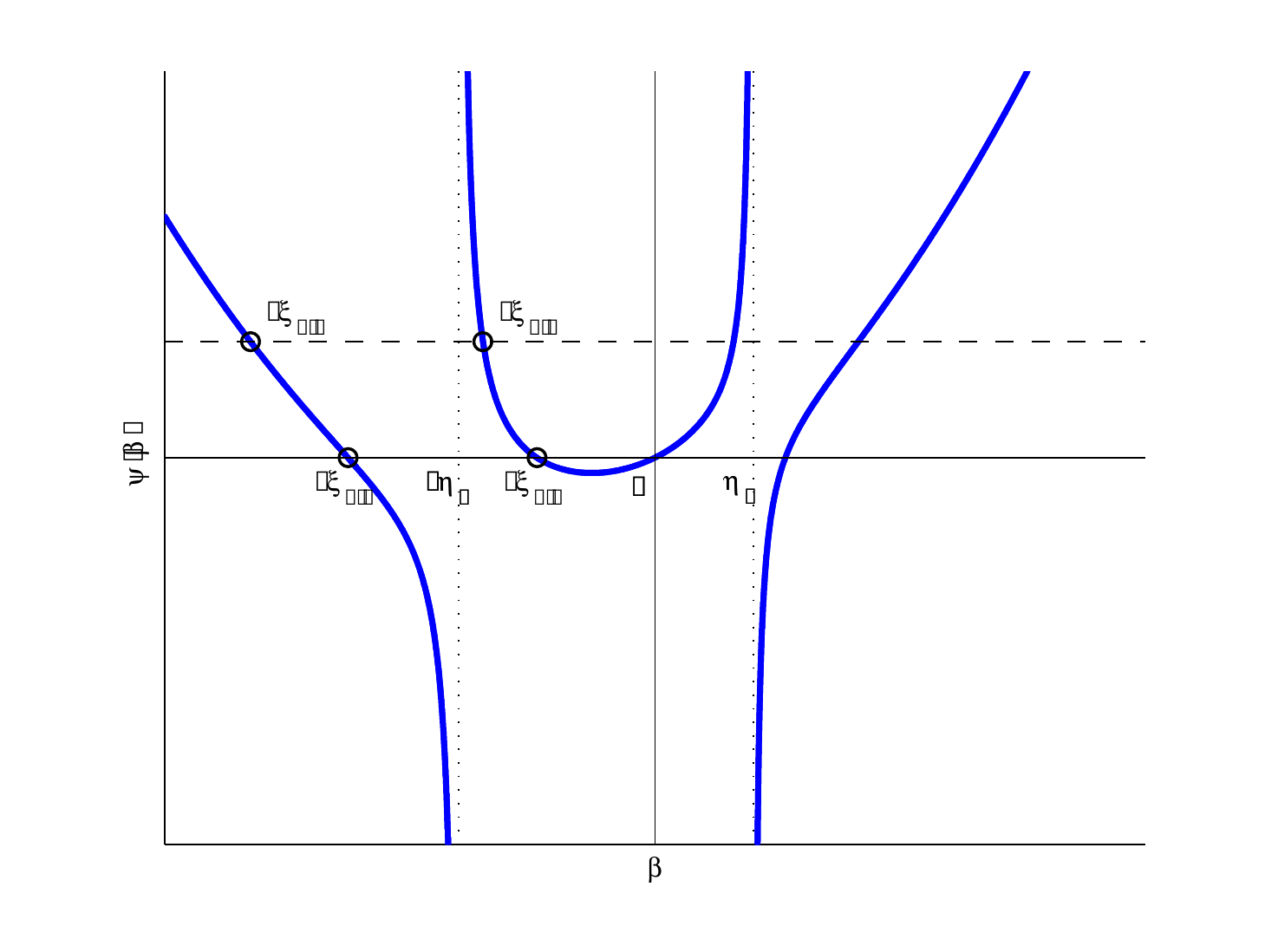} \\
(a) $\bar{u} < 0$ & (b) $\bar{u}>0$
\end{tabular}
\end{minipage}
\caption{Illustration of the Laplace exponent and its roots when the drift $\overline{u}$ is negative and positive.}

\label{fig:xi}
\end{center}
\end{figure}

Due to the memoryless
property of its jump-size distribution, the violation risk and regret can
be obtained explicitly.  The following two lemmas are due to Kou and Wang \cite{Kou_Wang_2003}, Theorem 3.1 and its corollary. Here we let
\begin{align*}
l_{1,q} := \frac {\eta_- - \xi_{1,q}} {\xi_{2,q} - \xi_{1,q}} > 0 \quad \textrm{and} \quad l_{2,q} := \frac {\xi_{2,q} - \eta_-} {\xi_{2,q} - \xi_{1,q}} > 0, \quad q \geq 0,
\end{align*}
where, in particular, when $q=0$ and $\overline{u} < 0$,
\begin{align*}
l_{1,0} = \frac {\eta_-} {\xi_{2,0}} > 0 \quad \textrm{and} \quad l_{2,0} = \frac {\xi_{2,0} - \eta_-} {\xi_{2,0}} > 0.
\end{align*}
Notice that $l_{1,q} + l_{2,q} = 1$ for every $q \geq 0$.
\begin{lemma}[violation risk] \label{proposition_undershoot}
For every $q \geq 0$ and $0 < A < x$, we have
\begin{align*}
R_x^{(q)}(\tau_A) = \frac {e^{-\eta_- A}} {\eta_-} \left[ (\xi_{2,q} - \eta_-)  l_{1,q} e^{-\xi_{1,q} (x-A)} -  (\eta_- - \xi_{1,q})  l_{2,q} e^{-\xi_{2,q} (x-A)} \right].
\end{align*}
In particular, when $q = 0$ and $\overline{u} < 0$, this reduces to
\begin{align*}
R_x^{(0)}(\tau_A) =l_{2,0}  e^{-\eta_- A}  \left[ 1 - e^{- \xi_{2,0}(x-A)}\right].
\end{align*}
\end{lemma}

\begin{lemma}[functional associated with the regret when $h\equiv1$] \label{proposition_expectation_tau_A}
For every $q > 0$, we have
\begin{align*}
\E^x \left[ \int_0^{\tau_A} e^{-q t} \diff t\right] =  \frac 1 {q \eta_-} \left[ l_{1,q} \xi_{2,q} \left( 1- e^{- \xi_{1,q}(x-A)} \right) + l_{2,q} \xi_{1,q} \left( 1 - e^{- \xi_{2,q}(x-A)} \right) \right], \quad 0 \leq A < x.
\end{align*}
Furthermore, it can be extended to the case $q = 0$, by taking $q \downarrow 0$ via (\ref{convergence_xi}) and the monotone convergence theorem;
\begin{align*}
\E^x \left[ \tau_A \right] = \left\{ \begin{array}{ll} \frac 1
{|\overline{u}|} \left[ (x-A) + \frac {\xi_{2,0}-\eta_-} {\eta_-
\xi_{2,0}} (1- e^{- \xi_{2,0}(x-A)}) \right], & \textrm{if } \;
\overline{u} < 0, \\ \infty, & \textrm{if } \; \overline{u} \geq
0.\end{array}\right.
\end{align*}
\end{lemma}
For a general \lev process, Lemma  \ref{proposition_expectation_tau_A} can be alternatively achieved by obtaining $\p^x \left\{ \underline{X}_{\textbf{e}_q} < A\right\}$ where $\textbf{e}_q$ is an independent exponential random variable with parameter $q \geq 0$ and $\underline{X}_t := \inf_{0 \leq u \leq t} X_u$ is the running infimum of $X$; see Bertoin \cite{Bertoin_1996}, Kyprianou \cite{Kyprianou_2006} or Chapter 2 of Surya \cite{Surya_2007}.  In particular, $\p^x \left\{ \underline{X}_{\textbf{e}_q} < A\right\}$ admits an analytical form when jumps are of phase-type (Asmussen et al.\ \cite{Asmussen_2004}); the above results can be seen as its special case.

\begin{remark}
As in Kou and Wang \cite{Kou_Wang_2003}, we assume throughout this section that $X$ contains a diffusion component $(\sigma > 0)$.  Here we do not consider the case $\sigma = 0$ because its spectrally negative case is covered in the next section.  The results for $\sigma = 0$ can be obtained similarly using the results by Asmussen et al.\ \cite{Asmussen_2004}.
\end{remark}

\subsection{Optimal strategy when $h \equiv 1$} We shall obtain the optimal solution for $q \geq 0$. When $q = 0$, we focus on the case when $\overline{u} < 0$ because
$\E^x \theta = \infty$ otherwise by Lemma \ref{proposition_expectation_tau_A}
and the problem becomes trivial as we discussed in Section \ref{sec:model}.

Suppose $q >0$.  By Lemma \ref{proposition_expectation_tau_A}, the stopping value (\ref{def_stopping_value}) becomes
\begin{align}\label{eq:G-expression}
G(x)= \gamma\ME\left[\int_0^\theta e^{-qt}\diff
t\right] = \sum_{i=1,2} C_{i,q} \left( 1 - e^{-\xi_{i,q}x}\right) = \frac \gamma q - \sum_{i=1,2} C_{i,q} e^{-\xi_{i,q}x}, \quad x > 0,
\end{align}
where
\begin{align}
C_{1,q} := \frac \gamma q \frac { l_{1,q} \xi_{2,q}} {\eta_-} \quad \textrm{and} \quad C_{2,q} := \frac \gamma q \frac { l_{2,q} \xi_{1,q}} {\eta_-}, \label{about_C}
\end{align}
which satisfy
\begin{align}
\sum_{i=1,2} C_{i,q} = \frac \gamma q \quad \textrm{and} \quad \sum_{i = 1,2}{C_{i,q}} \frac {\eta_-} {\eta_- - \xi_{i,q}} = \frac \gamma q. \label{about_C_2}
\end{align}
The difference between the continuation and stopping values (defined
in (\ref{diff_v_g})) becomes, by Lemmas \ref{proposition_undershoot}-\ref{proposition_expectation_tau_A},
\begin{multline}
\delta_A(x) = \frac {l_{1,q}} {\eta_-}  \left[  (\xi_{2,q}-\eta_-) e^{-\xi_{1,q} (x-A)-\eta_- A}
- \frac \gamma q \xi_{2,q} \left( 1 - e^{- \xi_{1,q}(x-A)}\right)\right] \\ + \frac {l_{2,q}} {\eta_-}  \left[  -(\eta_- - \xi_{1,q})  e^{-\xi_{2,q} (x-A) -\eta_- A}
- \frac \gamma q \xi_{1,q} \left( 1 - e^{- \xi_{2,q}(x-A)}\right)\right], \quad 0 < A <
x. \label{delta_A}
\end{multline}

Suppose $q=0$ and $\overline{u} < 0$. By Lemma \ref{proposition_expectation_tau_A}, we have
\begin{align*}
G (x) &= \frac \gamma {|\overline{u}|} \left[ x + \frac {\xi_{2,0} - \eta_-} {\eta_- \xi_{2,0}} \left( 1 - e^{- \xi_{2,0}x}\right)\right], \quad x > 0,
\end{align*}
and  taking $q \rightarrow
0$ in \eqref{delta_A} via the monotone convergence theorem (or by Lemmas \ref{proposition_undershoot}-\ref{proposition_expectation_tau_A})
\begin{multline}
\delta_A(x) = \frac {l_{1,0}} {\eta_-} \left[ (\xi_{2,0} - \eta_-) e^{-\eta_- A} - \frac \gamma {|\overline{u}|} \xi_{2,0}(x-A) \right]  \\ + \frac {l_{2,0}} {\eta_-} \left[ - \eta_- e^{-\xi_{2,0}(x-A) - \eta_- A} - \frac \gamma {|\overline{u}|} \left( 1 - e^{-\xi_{2,0}(x-A)}\right)\right], \quad 0 < A < x. \label{delta_A_zero}
\end{multline}

\begin{remark} \label{remark_cont_fit}\normalfont
We have  $\delta_A(A+) = 0$ for every $A > 0$, i.e., continuous fit
holds whatever the choice of $A$ is. This is due to the fact that $X$ has paths of unbounded variation ($\sigma > 0$).  As we see in the next section, continuous fit is applied to identify the optimal threshold level for the bounded variation case. \qed
\end{remark}

We shall obtain the threshold level $A^*$ such that the smooth fit condition holds, i.e., $\delta'_{A^*}(A^*+) = 0$ if such a threshold exists. By \eqref{delta_A}-\eqref{delta_A_zero}, we have 
\begin{align*}
\delta_A'(A+)= \left\{ \begin{array}{ll} \frac 1 {\eta_-}\left[  (\eta_- - \xi_{1,q})  (\xi_{2,q}-\eta_-)  e^{-\eta_- A}  - \frac \gamma q \xi_{1,q} \xi_{2,q} \right], & q > 0 \\   (\xi_{2,0}-\eta_-)  e^{-\eta_- A}  - \frac \gamma {|\overline{u}| \eta_-} \xi_{2,0}, & q=0 \; \textrm{and} \; \overline{u} < 0 \end{array} \right\}, \quad A > 0.
\end{align*}
Therefore, on condition that
\begin{align}
\left\{ \begin{array}{ll}  (\eta_- - \xi_{1,q})  (\xi_{2,q}-\eta_-) > \frac \gamma q \xi_{1,q} \xi_{2,q}, & q > 0 \\   (\xi_{2,0}-\eta_-)  > \frac \gamma {|\overline{u}| \eta_-} \xi_{2,0}, & q=0 \; \textrm{and} \; \overline{u} < 0 \end{array} \right\}, \label{condition_smooth_fit_exists}
\end{align}
the smooth fit condition $\delta'_{A}(A+) = 0$ is satisfied if and only if $A=A^* > 0$ where
\begin{align}
A^* := \left\{ \begin{array}{ll} - \frac 1 {\eta_-} \log \left(  \frac \gamma q \frac {\xi_{1,q} \xi_{2,q}} {  (\eta_- - \xi_{1,q})(\xi_{2,q}-\eta_-)} \right), & q > 0, \\ -
\frac 1 {\eta_-} \log \left( \frac {\gamma \xi_{2,0}} {|\overline{u}|
\eta_- (\xi_{2,0}-\eta_-) }\right), & q = 0 \; \textrm{ and } \; \overline{u} < 0. \end{array} \right. \label{def_a_tilde_q}
\end{align}
If (\ref{condition_smooth_fit_exists}) does not hold, we have $\delta_A'(A+) < 0$ for every $A > 0$; in this case, we set $A^* = 0$.

We now show that the optimal value function is $\phi := \phi_{A^*}$ (see \eqref{phi_A}).  Suppose $q > 0$ and $A^* > 0$. Simple algebra shows that
\begin{align}
\delta(x) :=  \delta_{A^*}(x) =  \frac \gamma q \frac 1 {\xi_{2,q} - \xi_{1,q}} \left[ {\xi_{2,q}} \left( e^{-\xi_{1,q} (x - {A^*}) }- 1 \right) - { \xi_{1,q}} \left( e^{-\xi_{2,q} (x - {A^*}) } - 1 \right) \right], \quad x > A^*. \label{delta}
\end{align}
This together with (\ref{eq:G-expression}) shows that
\begin{align*}
\phi(x) = \left\{ \begin{array}{ll} \sum_{i = 1,2} (L_{i,q} - C_{i,q}) e^{-\xi_{i,q}x}, & x > A^*, \\ \frac \gamma q - \sum_{i=1,2} C_{i,q} e^{-\xi_{i,q}x}, & 0 < x \leq A^*, \\ 1, & x \leq 0, \end{array} \right.
\end{align*}
where $C_{1,q}$ and  $C_{2,q}$ are defined in (\ref{about_C}) and
\begin{align}
L_{1,q} := \frac \gamma q \frac {\xi_{2,q}} {\xi_{2,q} - \xi_{1,q}}  e^{ \xi_{1,q}{A^*}} \quad \textrm{and} \quad L_{2,q} := - \frac \gamma q \frac {\xi_{1,q}} {\xi_{2,q} - \xi_{1,q}}  e^{ \xi_{2,q}{A^*}}. \label{def:L}
\end{align}
When $q = 0$ and $A^* > 0$, we have
\begin{align*}
\delta(x) &= -\frac \gamma {|\overline{u}|} \left[ (x-A^*) -  \frac {1-e^{- \xi_{2,0}(x-A^*)}} {\xi_{2,0}} \right], \quad x > A^*,
\end{align*}
and consequently, 
\begin{align*}
\phi(x) &= \left\{ \begin{array}{ll} \frac \gamma {|\overline{u}|} \left( A^* + \frac {\xi_{2,0} - \eta_-} {\eta_- \xi_{2,0}} \left( 1 - e^{- \xi_{2,0}x}\right)  +  \frac 1 {\xi_{2,0}} (1-e^{- \xi_{2,0}(x-A^*)})   \right),  & x > A^*, \\
\frac \gamma {|\overline{u}|} \left( x + \frac {\xi_{2,0} - \eta_-} {\eta_- \xi_{2,0}} \left( 1 - e^{- \xi_{2,0}x}\right)\right), & 0 < x \leq A^*, \\ 1, & x \leq 0. \end{array} \right.
\end{align*}
%
%
%

Finally, it is understood for the case $A^*=0$ (for both $q > 0$ and $q=0$) that
\begin{align}
\phi(x) = \lim_{\varepsilon \downarrow 0}\phi_\varepsilon (x) 
 \label{v_negative_case}
\end{align}
where, by Lemma \ref{proposition_undershoot},
\begin{align*}
\lim_{\varepsilon \downarrow 0} R_x^{(q)} (\tau_\varepsilon)= \E^x \left[ e^{- q \theta}  1_{\{ X_{\theta} < 0, \, \theta < \infty \}} \right] =  \frac {(\eta_- - \xi_{1,q})(\xi_{2,q}-\eta_-)} {\eta_- (\xi_{2,q} - \xi_{1,q})} \left( e^{-\xi_{1,q} x} - e^{-\xi_{2,q} x} \right).
\end{align*}
This is the expectation of the cost incurred only when it jumps over the level zero.

When $A^* > 0$ the value function $\phi$ can be attained by $\tau_{A^*}$, and when $A^* = 0$ it can be approximated arbitrarily closely by $\phi_\varepsilon(x)$ (which is attained by $\tau_\varepsilon$) for sufficiently small $\varepsilon > 0$.  We therefore only need to verify that $\phi(x) \leq \E^x \left[ e^{- q\tau} G(X_\tau) 1_{\{\tau < \infty\}}\right]$ for any $\tau \in \S$.

We first show that  $\phi(\cdot)$
is dominated from above by the stopping value $G(\cdot)$.
\begin{lemma} \label{lemma_difference}
We have $\phi(x) \leq G(x)$ for every $x \in \R$.
\end{lemma}
\begin{proof}
For every $0 < A < x$,
${\partial} \phi_A(x) / \partial A = {\partial} \delta_A(x) / \partial A$ equals
\begin{align*}
 \left\{ \begin{array}{ll} \frac 1 {\eta_-}\left[l_{1,q} e^{- \xi_{1,q}(x-A)} +
{l_{2,q}} e^{- \xi_{2,q}(x-A)} \right] \left[ -(\eta_- -\xi_{1,q}) (\xi_{2,q} - \eta_-) e^{-\eta_- A} + {\frac \gamma q}\xi_{1,q} \xi_{2,q} \right], & q > 0 \\ \frac 1 {\eta_-} \left[ l_{1,0} + l_{2,0} e^{- \xi_{2,0}(x-A)} \right] \left[ -\eta_- (\xi_{2,0} - \eta_-) e^{-\eta_- A} + \frac \gamma {|\overline{u}|} \xi_{2,0} \right], & q=0 \; \textrm{and} \; \overline{u} < 0 \end{array} \right\}.
\end{align*}
Here, in both cases, the term in the first bracket is strictly positive while that in the second bracket is increasing in $A$. Therefore,  when $A^* > 0$, $A^*$ is the unique value that makes it vanish, and consequently  ${\partial} \phi_A(x) / \partial A \geq 0$ if and only if $A \geq A^*$.  On the other hand, if $A^* = 0$,  ${\partial} \phi_A(x) / \partial A \geq 0$ for every $0 < A < x$. These imply when $x > A^*$ that by Remark \ref{remark_cont_fit}
\begin{align*}
\phi(x)  \leq \lim_{A \uparrow x}\phi_{A} (x) = G(x) \quad \textrm{and} \quad
\phi(x) = \lim_{\varepsilon \downarrow 0}\phi_\varepsilon(x) \leq \lim_{A \uparrow x}\phi_{A} (x) = G(x)
\end{align*}
when $A^* > 0$ and when $A^* = 0$, respectively.  On the other hand, when $-\infty < x \leq A^*$, we have $\phi(x) = G(x)$ by definition and hence the proof is complete.
\end{proof}

\begin{remark} \label{remark:phi_bounded}\normalfont
\begin{enumerate}
\item Suppose $q > 0$.  In view of (\ref{eq:G-expression}), $G (\cdot)$ is bounded  from above by $\gamma /q$ uniformly on $x \in (0,\infty)$.
\item Suppose $\overline{u} < 0$ and fix $x > 0$.  We have $G(x) \leq \E^x \theta < \infty$ uniformly on $q \geq 0$; namely, $G(x)$ is bounded by $\E^x \theta < \infty$ uniformly on $q \in [0,\infty)$.
\item Suppose $A^* > 0$. Using the same argument as in the proof of Lemma \ref{lemma_difference}, we have $\phi(x)  \leq \lim_{\varepsilon \downarrow 0}\phi_\varepsilon(x) = \E^x \left[ e^{- q \theta}  1_{\{ X_{\theta} < 0, \, \theta < \infty \}} \right] \leq 1$ for every $x > A^*$. When $0 < x \leq A^*$, using the monotonicity of $G(\cdot)$ and continuous fit (see Remark \ref{remark_cont_fit}), we have
$\phi(x) = G(x) \leq G(A^*) = \phi(A^*) \leq \lim_{\varepsilon \downarrow 0}\phi_\varepsilon(A^*)  \leq 1$.  When $A^*=0$, we have $\phi(x) \leq 1$ in view of \eqref{v_negative_case}.
Therefore, $\phi(\cdot)$ is uniformly bounded and it only takes values on
$[0,1]$. \qed
\end{enumerate}
\end{remark}

Now we show that $\mathcal{L} \phi(x) \geq q \phi(x)$ on $(0,\infty) \backslash \{A^*\}$ where $\mathcal{L}$ is the \emph{infinitesimal generator} of $X$ such that
\begin{align}
\mathcal{L} w(x) = \mu w' (x) + \frac 1 2 \sigma^2 w''(x) +  \lambda \int_{-\infty}^\infty \left[ w(x+z) - w(x) \right] f (z) \diff z, \quad x \in \mathbb{R}, \label{generator_double_exponential}
\end{align}
for any $C^2$-function $w: \R \rightarrow \R$. The proof for the following lemma is lengthy and technical, and therefore is relegated to the appendix.

\begin{lemma} \label{lemma_verification}
\begin{enumerate}
\item If $A^* > 0$, then we have
\begin{align}
\mathcal{L} \phi(x) - q \phi(x) &= 0, \quad x > A^*, \label{equality_continuation}\\
\mathcal{L} \phi(x) - q \phi(x) &> 0, \quad 0 < x < A^*. \label{inequality_stopping}
\end{align}
\item If $A^* = 0$, then (\ref{equality_continuation}) holds for every $x > 0$.
\end{enumerate}
\end{lemma}

Lemmas \ref{lemma_difference} and \ref{lemma_verification} show the optimality. The proof is very similar to that of Proposition \ref{proposition_optimality_spectrally_negative} given in the next section (see Appendix \ref{proof_proposition_optimality_spectrally_negative}) and hence we omit it.
\begin{proposition} \label{proposition_optimality_case_1}
We have
\begin{align*}
\phi(x) = \inf_{\tau \in \S} \E^x \left[ e^{-q \tau} G(X_\tau) 1_{\left\{ \tau < \infty \right\}} \right], \quad x > 0.
\end{align*}
\end{proposition}


\section{Spectrally Negative Case}  \label{sec:spectrally_negative}
In this section, we analyze the case for a general spectrally negative \lev process. We shall obtain the optimal strategy and the value function in terms of the \emph{scale function} for a general $h$.   We assume throughout this section that $q > 0$.  The results obtained here can in principle be extended to the case $q=0$ on condition that $\E^x \theta < \infty$ along the same line as in the discussion in the previous section.  The proofs of all lemmas and propositions are given in the appendix.

\subsection{Scale functions} \label{subsec:scale_functions}

Let $X$ be a spectrally negative \lev process with  its Laplace exponent
\begin{align*}
\psi(\beta) := \E^0 \left[ e^{\beta X_1}\right] = c \beta +\frac{1}{2}\sigma^2 \beta^2 +\int_{(
0,\infty)}(e^{-\beta x}-1+\beta x 1_{\{0 <x<1\}})\,\Pi(\diff x), \quad  {\beta \in \mathbb{R}},
\end{align*}
where $c \in\R$, $\sigma\geq 0$ and $\Pi$ is a measure  on $(0,\infty)$ such that \mbox{$\int_{(0,\infty)} (1  \wedge x^2) \Pi( \diff x)<\infty$}. See Kyprianou \cite{Kyprianou_2006}, p.212.  In particular, when
\begin{align}
\int_{(0,\infty)} (1 \wedge x)\, \Pi(\diff x) < \infty,  \label{cond_bounded_variation}
\end{align}
we can rewrite
\begin{align*}
\psi(\beta) =\mu  \beta+  \frac{1}{2}\sigma^2 \beta^2 + \int_{(
0,\infty)}(e^{-\beta x}-1)\,\Pi(\diff x), \quad \beta \in \mathbb{R}
\end{align*}
where
\begin{align}
\mu := c + \int_{(
0,1)}x\, \Pi(\diff x). \label{def_mu}
\end{align}
The process has paths of bounded variation if and only if $\sigma = 0$ and \eqref{cond_bounded_variation} holds. It is also assumed that $X$ is not a negative subordinator (decreasing a.s.). Namely, we require $\mu$ to be strictly positive if $\sigma = 0$.

It is well-known that $\psi$ is zero at the origin, convex on $\R_+$ and has a right-continuous
inverse:
\[
\zeta_q :=\sup\{\lambda \geq 0: \psi(\lambda)=q\}, \quad
q\ge 0.
\]
Associated with every spectrally negative \lev process, there exists a \emph{(q-)scale
function}
\begin{align*}
W^{(q)}: \R \rightarrow \R; \quad q\ge 0,
\end{align*}
that is continuous and strictly increasing on $[0,\infty)$ and satisfies
\begin{align*}
\int_0^\infty e^{-\beta x} W^{(q)}(x) \diff x = \frac 1
{\psi(\beta)-q}, \qquad \beta > \zeta_q.
\end{align*}
If $\tau_a^+$ is the first time the process goes above $a > x > 0$, we have
\begin{align*}
\E^x \left[ e^{-q \tau_a^+} 1_{\left\{ \tau_a^+ < \theta, \, \tau_a^+ < \infty \right\}}\right] = \frac {W^{(q)}(x)}  {W^{(q)}(a)} \quad \textrm{and}  \quad \E^x \left[ e^{-q \theta} 1_{\left\{ \tau_a^+ > \theta, \, \theta < \infty \right\}}\right] = Z^{(q)}(x) - Z^{(q)}(a) \frac {W^{(q)}(x)} {W^{(q)}(a)},
\end{align*}
where
\begin{align}\label{eq:Z-q-function}
Z^{(q)}(x) := 1 + q \int_0^x W^{(q)}(y) \diff y, \quad x \in \R.
\end{align}
Here we have
\begin{equation}\label{eq:at-zero}
W^{(q)}(x)=0 \quad\text{on} \quad  (-\infty,0)\quad\text{ and}\quad
Z^{(q)}(x)=1 \quad\text{on}\quad (-\infty,0].
\end{equation}
We assume that $\Pi$ does not have atoms; this implies that $W^{(q)}$ is continuously differentiable on $(0,\infty)$.  See Chan et al.\ \cite{Chan_2009} for the smoothness properties of the scale function.

The scale function increases exponentially; indeed, we have
\begin{align}
W^{(q)} (x) \sim \frac {e^{\zeta_q x}} {\psi'(\zeta_q)} \quad \textrm{as } \; x \rightarrow \infty. \label{scale_function_asymptotic}
\end{align}
There exists a (scaled) version of the scale function
$ W_{\zeta_q} = \{ W_{\zeta_q} (x); x \in \R \}$ that satisfies, for every fixed $q
\geq 0$,
\begin{align*}
W_{\zeta_q} (x) = e^{-\zeta_q x} W^{(q)} (x), \quad x \in \R
\end{align*}
and 
\begin{align*}
\int_0^\infty e^{-\beta x} W_{\zeta_q} (x) \diff x &= \frac 1 {\psi(\beta+\zeta_q)-q}, \quad \beta > 0.
\end{align*}
Moreover $W_{\zeta_q} (x)$ is increasing and as is clear from \eqref{scale_function_asymptotic}
\begin{align}
W_{\zeta_q} (x) \nearrow \frac 1 {\psi'(\zeta_q)} \quad \textrm{as } \; x \rightarrow \infty. \label{scale_function_asymptotic_version}
\end{align}

From Lemmas 4.3 and 4.4 of Kyprianou and Surya \cite{Kyprianou_Surya_2007}, we also have the following results about the behavior in the neighborhood of zero:
\begin{align}
W^{(q)} (0) = \left\{ \begin{array}{ll} 0, & \textrm{unbounded
variation} \\ \frac 1 {\mu}, & \textrm{bounded variation}
\end{array} \right\} \quad \textrm{and} \quad W^{(q)'} (0+) =
\left\{ \begin{array}{ll}  \frac 2 {\sigma^2}, & \sigma > 0 \\
\infty, & \sigma = 0 \; \textrm{and} \; \Pi(0,\infty) = \infty \\
\frac {q + \Pi(0,\infty)} {\mu^2}, & \textrm{compound Poisson}
\end{array} \right\}. \label{at_zero}
\end{align}

For a comprehensive account of the scale function, see Bertoin \cite{Bertoin_1996,Bertoin_1997}, Kyprianou \cite{Kyprianou_2006} and
Kyprianou and Surya \cite{Kyprianou_Surya_2007}. See Surya \cite{Surya_2008} and Egami and Yamazaki  \cite{Egami_Yamazaki_2010_2} for numerical methods for computing the scale function.

\subsection{Rewriting the problem in terms of the scale function for a general $h$} \label{subsec:penalty_function}
We now rewrite the problem in terms of the scale function. For fixed $A \geq 0$, define a random measure
\begin{align*}
M^{(A,q)}(\omega, B) := \int_0^{\tau_{A}(\omega)} e^{-qt} 1_{\left\{
X_t(\omega) \in B \right\}} \diff t, \quad \omega \in \Omega, \; B
\in \mathcal{B}(\R).
\end{align*}
\begin{lemma} \label{lemma_measure}
For any $\omega \in \Omega$, we have
\begin{align}\label{eq:M-integral}
\int_0^{\tau_A(\omega)} e^{-qt} h(X_t(\omega))\diff t = \int_{\mathbb{R}} M^{(A,q)} (\omega,\diff y) h(y).
\end{align}
\end{lemma}

With this lemma and the property of the random measure, we have
\begin{align}\label{eq:mean-measure}
\E^x \left[ \int_0^{\tau_{A}} e^{-qt}  h(X_t) \diff t \right] = \E^x
\left[ \int_{\mathbb{R}} M^{(A,q)} (\omega,\diff y) h(y) \right]
=\int_{\mathbb{R}} \mu_x^{(A,q)}(\diff y) h(y)
\end{align}
where
\begin{align*}
\mu_x^{(A,q)} (B) := \E^x \left[ M^{(A,q)}(B) \right], \quad B \in \mathcal{B} (\mathbb{R})
\end{align*}
is a version of the \emph{q-resolvent kernel} that has a density owing to the Radon-Nikodym theorem; see Bertoin \cite{Bertoin_1997}.  By using Theorem 1 of Bertoin \cite{Bertoin_1997} (see also Emery \cite{Emery_1973} and Suprun \cite{Suprun_1976}),  we have for every $B \in \mathcal{B}(\mathbb{R})$ and $a > x$
\begin{align*}
\E^x \left[ \int_0^{\tau_{A} \wedge \tau^+_a} e^{-qt} 1_{\left\{ X_t \in B \right\}} \diff t\right] &= \int_{B \cap [A,\infty)} \Big[ \frac {W^{(q)}(x-A) W^{(q)} (a-y)} {W^{(q)}(a-A)} - 1_{\{ x \geq y \}} W^{(q)} (x-y) \Big] \diff y \\
&= \int_{B \cap [A,\infty)} \Big[ e^{-\zeta_q(y-A)} \frac {W^{(q)}(x-A) W_{\zeta_q} (a-y)} {W_{\zeta_q}(a-A)} - 1_{\{ x \geq y \}} W^{(q)} (x-y) \Big] \diff y.
\end{align*}
Moreover, by taking $a \uparrow \infty$ via the dominated convergence theorem in view of \eqref{scale_function_asymptotic_version}, we have
\begin{align*}
\mu_x^{(A,q)} (B) &= \int_{B \cap [A,\infty)} \lim_{a \rightarrow \infty}\Big[ e^{-\zeta_q(y-A)} \frac {W^{(q)}(x-A) W_{\zeta_q} (a-y)} {W_{\zeta_q}(a-A)} - 1_{\{ x \geq y \}} W^{(q)} (x-y) \Big] \diff y \\
&= \int_{B \cap [A,\infty)} \left[ e^{- \zeta_q(y-A)} W^{(q)} (x-A) - 1_{\{x \geq y\}}W^{(q)} (x-y) \right] \diff y
\end{align*}
where the second equality holds by (\ref{scale_function_asymptotic_version}). Hence, we have the following result.
\begin{lemma}\label{lem:3} Fix $q > 0$ and $0 < A < x$. We have
\begin{align*}
\mu_x^{(A,q)}(\diff y) = \left\{ \begin{array}{ll} \Big(e^{- \zeta_q(y-A)} W^{(q)} (x-A) - 1_{\{x \geq y\}} W^{(q)} (x-y) \Big) \diff y, & y \geq A, \\
0, & y < A.
\end{array} \right.
\end{align*}
\end{lemma}

By  (\ref{eq:mean-measure}) and Lemma \ref{lem:3}, we have, for any arbitrary $0 < A < x$,
\begin{align}
\E^{x}\left[ \int_0^{\tau_{A}} e^{-qt} h(X_t) \diff t \right] =
W^{(q)} (x-A) \int_0^\infty e^{- \zeta_q y}   h(y+A) \diff y -
\int_A^x W^{(q)} (x-y) h(y) \diff y, \label{eq:integral_to_tau_A}
\end{align}
and this can be used to express the regret function and the stopping value.  This also implies that the integrability condition \eqref{cond_integrability} is equivalent to $\int_0^\infty e^{- \zeta_q y}   h(y) \diff y < \infty$.

Using the $q$-resolvent kernel, we can also rewrite the violation risk.
\begin{lemma} \label{violation_risk_spectrally_negative}
For every $0 < A < x$, we have
\begin{align*}
R^{(q)}_x (\tau_A) = \frac 1 {\zeta_q} W^{(q)}(x-A) \int_A^\infty \Pi(\diff u) \left( 1 - e^{-\zeta_q (u-A)}\right) - \frac 1 q \int_A^\infty \Pi(\diff u) \left( Z^{(q)}(x-A)-Z^{(q)}(x-u)\right).
\end{align*}
\end{lemma}

By \eqref{eq:integral_to_tau_A} and Lemma \ref{violation_risk_spectrally_negative}, the difference function \eqref{diff_v_g} becomes, for all $0 < A < x$,
\begin{multline}
\begin{split}
\delta_A(x) &= R_x^{(q)}(\tau_{A})  -
\gamma \E^x \left[ \int^{\tau_A}_0 e^{-qt} h(X_t) \diff t \right] \\ &=\frac 1 {\zeta_q} W^{(q)}(x-A) \int_A^\infty \Pi(\diff
u) \left( 1 - e^{-\zeta_q (u-A)}\right) - \frac 1 q \int_A^\infty
\Pi(\diff u) \left( Z^{(q)}(x-A)-Z^{(q)}(x-u)\right) \\ & \qquad - \gamma
\left[ W^{(q)} (x-A) \int_0^\infty e^{- \zeta_q y}   h(y+A) \diff y
- \int_A^x W^{(q)} (x-y) h(y) \diff y \right].
\end{split}
\label{delta_spec_negative}
\end{multline}

\subsection{Continuous and smooth fit}
We now apply continuous and smooth fit to obtain the candidate threshold levels for the bounded and unbounded variation cases, respectively. Firstly, the continuous fit condition $\delta_A(A+) = 0$ requires in view of \eqref{delta_spec_negative} that
\begin{align}
W^{(q)}(0)  \Phi(A) = 0 \label{continuous_fit_cond}
\end{align}
where
\begin{align*}
\Phi(A) := \frac 1 {\zeta_q} \int_A^\infty \Pi(\diff u) \left( 1 - e^{-\zeta_q (u-A)}\right) - \gamma  \int_0^\infty e^{- \zeta_q y}   h(y+A) \diff y, \quad A > 0.
\end{align*}
Note in this calculation we used the fact that the second term in
(\ref{delta_spec_negative}) vanishes as $x \downarrow A$ by (\ref{eq:at-zero}). The condition \eqref{continuous_fit_cond}
automatically holds for the unbounded variation case by
\eqref{at_zero}, but for the bounded variation case it requires
\begin{align}
\Phi(A)  = 0.
\label{condition_spectrally_negative} \end{align}

For the unbounded variation case, we apply smooth fit.
For the violation risk, by using \eqref{eq:Z-q-function}-\eqref{eq:at-zero}, we have
\begin{align*}
\left. \frac \partial {\partial x} R^{(q)}_x (\tau_A) \right|_{x=A+}
&= \frac 1 {\zeta_q} W^{(q)'}(0+) \int_A^\infty \Pi(\diff u) \left(
1 - e^{-\zeta_q (u-A)}\right).
\end{align*}
For the regret function, by using \eqref{at_zero} in
particular $W^{(q)}(0)=0$, we obtain
\begin{align*}
 \left. \frac {\partial} {\partial x} \E^{x}\left[
\int_0^{\tau_{A}} e^{-qt} h(X_t) \diff t \right]  \right|_{x = A+}
&= \lim_{x \downarrow A} \left[ W^{(q)'} (x-A) \int_0^\infty e^{-
\zeta_q y}
h(y+A) \diff y  - \int_A^x  W^{(q)'} (x-y) h(y) \diff y \right] \\
&= W^{(q)'} (0+) \int_0^\infty e^{- \zeta_q y}  h(y+A) \diff y .
\end{align*}
Therefore, for the unbounded variation case, smooth fit requires
\begin{align*}
W^{(q)'}(0+) \Phi(A) = 0.
\end{align*}

Consequently, once we get \eqref{condition_spectrally_negative},
continuous fit holds for the bounded variation case and both
continuous and smooth fit holds for the unbounded variation case
(see Figure \ref{fig:smooth_continuous_fit} in Section
\ref{sec:numerical_results} for an illustration).  Because $h$ is
non-decreasing by assumption, we have
\begin{align*}
\Phi'(A) = - \int_A^\infty \Pi(\diff u)e^{-\zeta_q (u-A)} - \gamma
\int_0^\infty e^{-\zeta_q y} h'(y+A) \diff y < 0, \quad A > 0.
\end{align*}
Hence there exists \emph{at most} one root that satisfies
\eqref{condition_spectrally_negative}. We let $A^*$ be the root if
it exists and zero otherwise.  Because $\lim_{A \uparrow \infty}\Phi(A) < 0$, $A^*=0$ means that $\Phi(A) < 0$ for all $A > 0$.

\subsection{Verification of optimality}
We now show as in the last section that the optimal value function is $\phi := \phi_{A^*}$ (see \eqref{phi_A}) where in particular the case $A^*=0$ is defined by \eqref{v_negative_case}.  When $A^* > 0$, it can be attained by the strategy $\tau_{A^*}$ while when $A^*=0$ it can be approximated arbitrarily closely by $\tau_\varepsilon$ with sufficiently small $\varepsilon > 0$. Recall that $\phi(x)=\delta(x)+G(x)$ (with $\delta := \delta_{A^*}$) for $x>A^*$ and that
$G(x)$ in \eqref{def_stopping_value} for $x>0$ can be expressed in
terms of the scale function by taking the limit $A\downarrow 0$ in
\eqref{eq:integral_to_tau_A}.  The corresponding candidate value function becomes both for $A^* = 0$ and for $A^* > 0$
\begin{align}
\label{phi_spectrally_negative_positive}
\begin{split}
\phi(x) &= \frac 1 {\zeta_q} W^{(q)}(x-A^*) \int_{A^*}^\infty \Pi(\diff u) \left( 1 - e^{-\zeta_q (u-A^*)}\right) - \frac 1 q \int_{A^*}^\infty \Pi(\diff u) \left( Z^{(q)}(x-A^*) - Z^{(q)} (x-u)\right)  \\ &+ \gamma \left[  W^{(q)}(x) \int_0^\infty e^{-\zeta_q y} h(y) \diff y  - W^{(q)}(x-A^*) \int_0^\infty e^{-\zeta_q y} h(y+A^*) \diff y - \int_0^{A^*} W^{(q)}(x-y) h(y) \diff y \right]
\end{split}
\end{align}
 for every $x > 0$.  By definition, $\phi(x)=1$ for every $x \leq 0$.
In particular, when $A^* > 0$, we can simplify by using
\eqref{condition_spectrally_negative},
\begin{multline*}
\phi(x) = - \frac 1 q \int_{A^*}^\infty \Pi(\diff u) \left( Z^{(q)}(x-A^*) - Z^{(q)} (x-u)\right) \\  + \gamma \left[  W^{(q)}(x) \int_0^\infty e^{-\zeta_q y} h(y) \diff y   - \int_0^{A^*} W^{(q)}(x-y) h(y) \diff y \right], \quad x > 0. 
\end{multline*}
These expressions for the candidate value function $\phi$ are
valid not only on $(A^*,\infty)$ but also on $(0,A^*]$ thanks to
\eqref{eq:at-zero} and continuous fit.\\

In order to verify that $\phi$ is indeed optimal, we
only need to show that (1) $\phi$ is dominated by $G$ and (2) $\mathcal{L} \phi(x) \geq q \phi(x)$ on $(0,\infty) \backslash \{A^*\}$ where 
\begin{align*}
\mathcal{L} f(x) &= c f'(x) + \frac 1 2 \sigma^2 f''(x) + \int_0^\infty \left[ f(x-z) - f(x) +  f'(x) z 1_{\{0 < z < 1\}} \right] \Pi(\diff z), \\
\mathcal{L} f(x) &= \mu f'(x) + \int_0^\infty \left[ f(x-z) - f(x) \right] \Pi(\diff z),
\end{align*}
for the unbounded and bounded variation cases, respectively; see \eqref{def_mu} for the definition of $\mu$.  The former is proved in the
following lemma:

\begin{lemma} \label{lemma_difference_spectrally_negative}
We have $\phi(x) \leq G(x)$ for every $x \in \R$.
\end{lemma}

Now we recall that the processes $\left\{ e^{-q (t \wedge \theta
\wedge \tau_a^+)} W^{(q)}(X_{t \wedge \theta \wedge \tau_a^+}); t \geq 0
\right\}$ and $\left\{ e^{-q (t \wedge \theta \wedge \tau_a^+)}
Z^{(q)}(X_{t \wedge \theta \wedge \tau_a^+}); t \geq 0 \right\}$ are
$\p^x$-martingales for any $0 < x < a$; see page 229 in Kyprianou
\cite{Kyprianou_2006}. Therefore, we have
\begin{align}
\mathcal{L} W^{(q)} (x) = q W^{(q)} (x) \quad \textrm{and} \quad \mathcal{L} Z^{(q)} (x) = q Z^{(q)} (x), \quad x > 0. \label{scale_generator_zero}
\end{align}
  We take advantage of \eqref{scale_generator_zero} to show the following lemma.


\begin{lemma} \label{generator_smaller_than_zero_spectrally_negative}
\begin{enumerate}
\item We have
\begin{align*}
\mathcal{L} \phi(x) - q \phi(x) = 0, \quad x > A^*.
\end{align*}
\item If $A^* > 0$, we have
\begin{align*}
\mathcal{L} \phi(x) - q \phi(x) > 0, \quad 0 < x < A^*.
\end{align*}
\end{enumerate}
\end{lemma}

Finally, Lemmas \ref{lemma_difference_spectrally_negative}-\ref{generator_smaller_than_zero_spectrally_negative} are used to show the optimality of  $\phi$.
\begin{proposition} \label{proposition_optimality_spectrally_negative}
We have
\begin{align*}
\phi(x) = \inf_{\tau \in \S} \E^x \left[ e^{-q \tau} G(X_\tau) 1_{\left\{ \tau < \infty \right\}} \right], \quad x > 0.
\end{align*}
\end{proposition}

\section{Numerical Results} \label{sec:numerical_results}

We conclude this paper by providing numerical results on the models studied in Sections \ref{sec:double_exponential}-\ref{sec:spectrally_negative}.  We obtain optimal threshold levels $A^*$ for (1) the double exponential jump diffusion case with $h \equiv 1$, and for (2) the spectrally negative case with $h$ in the form of the \emph{exponential utility function}. We study how the solution depends on the process $X$.  We then verify continuous and smooth fit conditions for the bounded and unbounded variation cases, respectively.

\subsection{The double exponential jump diffusion case with $h \equiv 1$}
\label{subsection_double_exponential} We evaluate the results
obtained in Section \ref{sec:double_exponential} focusing on the
case $h \equiv 1$.  Here we plot the optimal threshold level $A^*$
defined in (\ref{def_a_tilde_q}) as a function of $\gamma$. The
values of $\xi_{1,q}$ and $\xi_{2,q}$ are obtained via the bisection
method with error bound $10^{-4}$.
\begin{figure}[htbp]
\begin{center}
\begin{minipage}{1.0\textwidth}
\centering
\begin{tabular}{cc}
\includegraphics[scale=0.5]{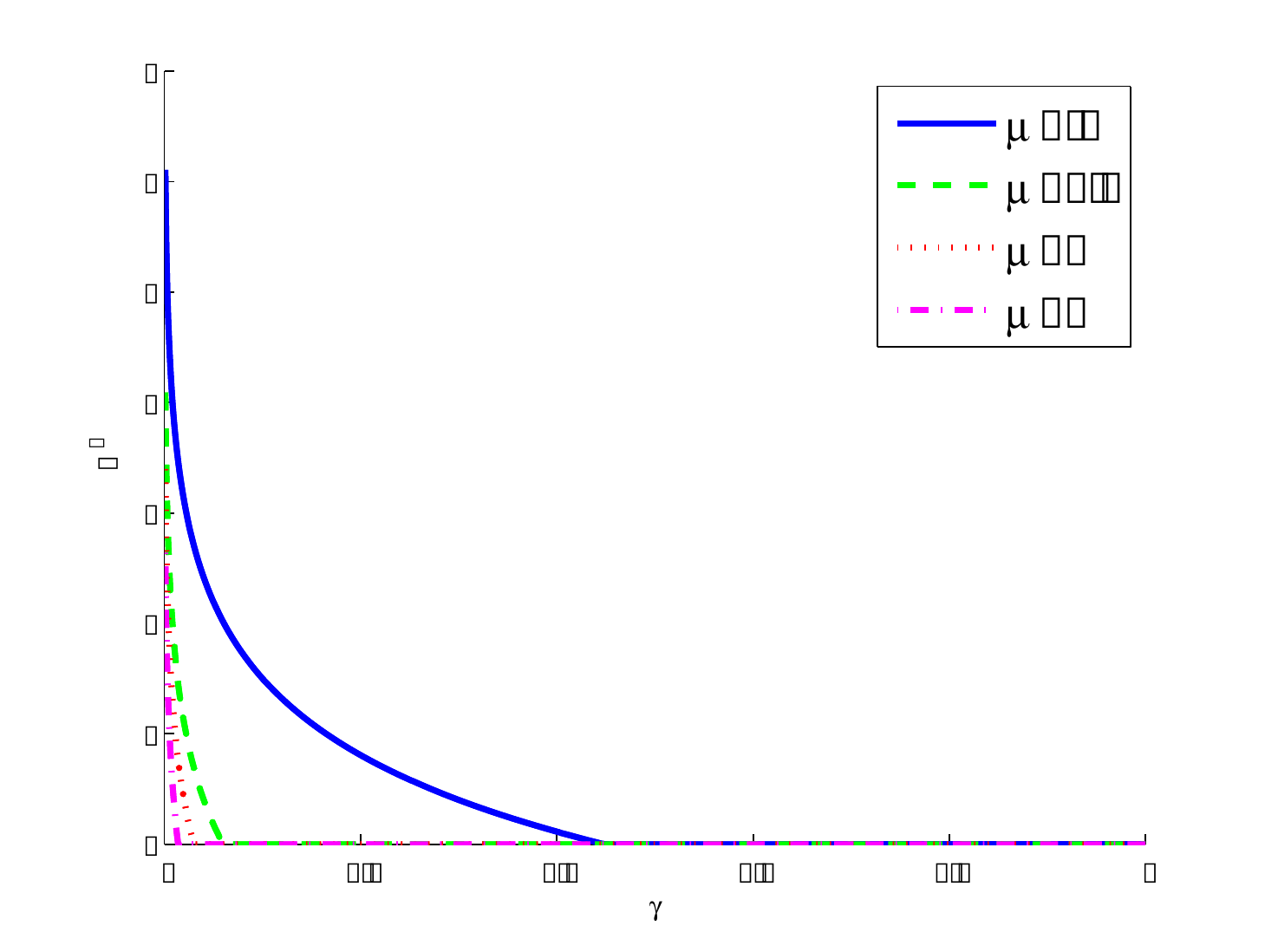}  & \includegraphics[scale=0.5]{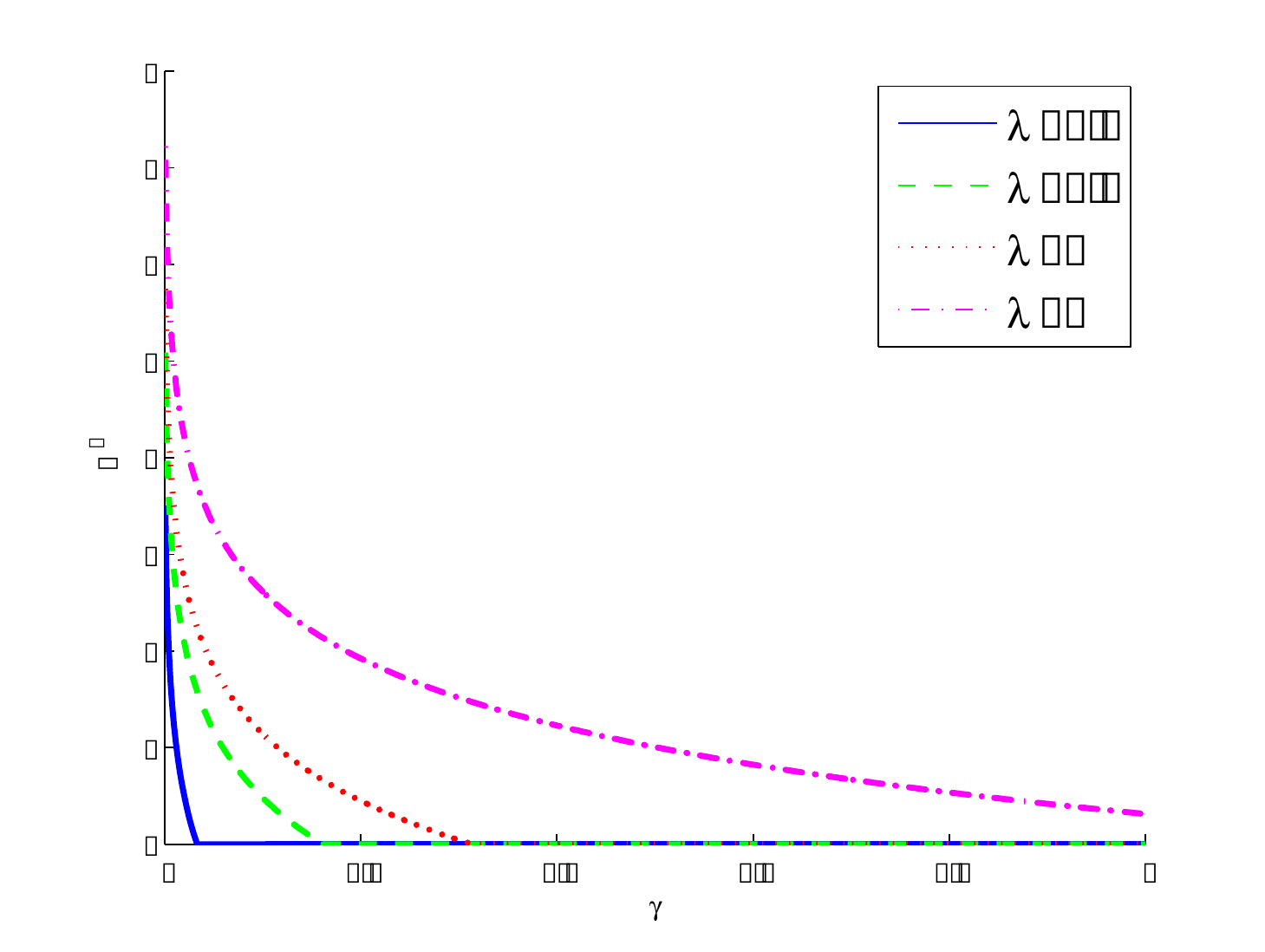}  \\
(i) $\mu$  & (ii) $\lambda$
 \vspace{0.7cm}
\\
\includegraphics[scale=0.5]{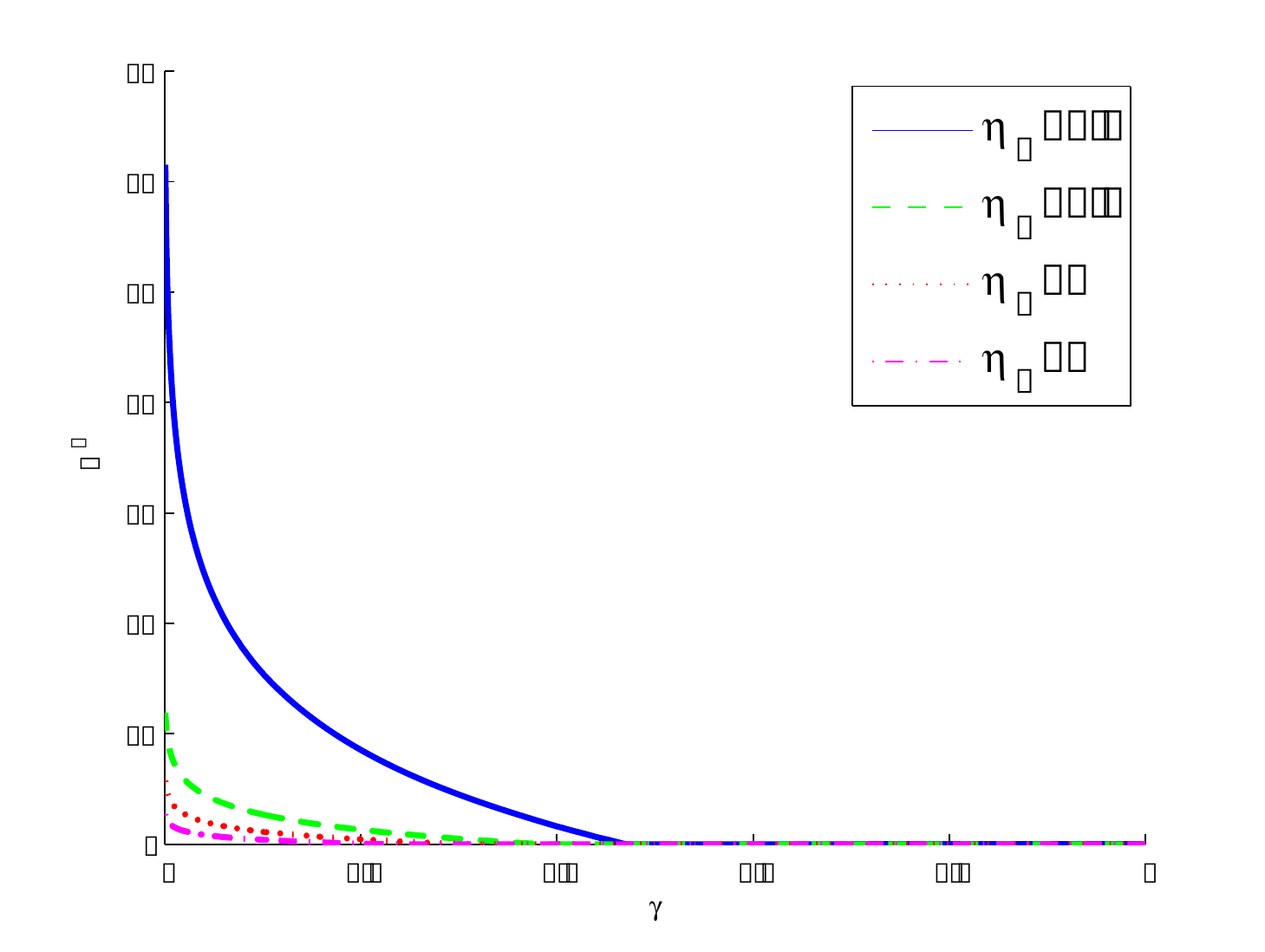}  & \includegraphics[scale=0.5]{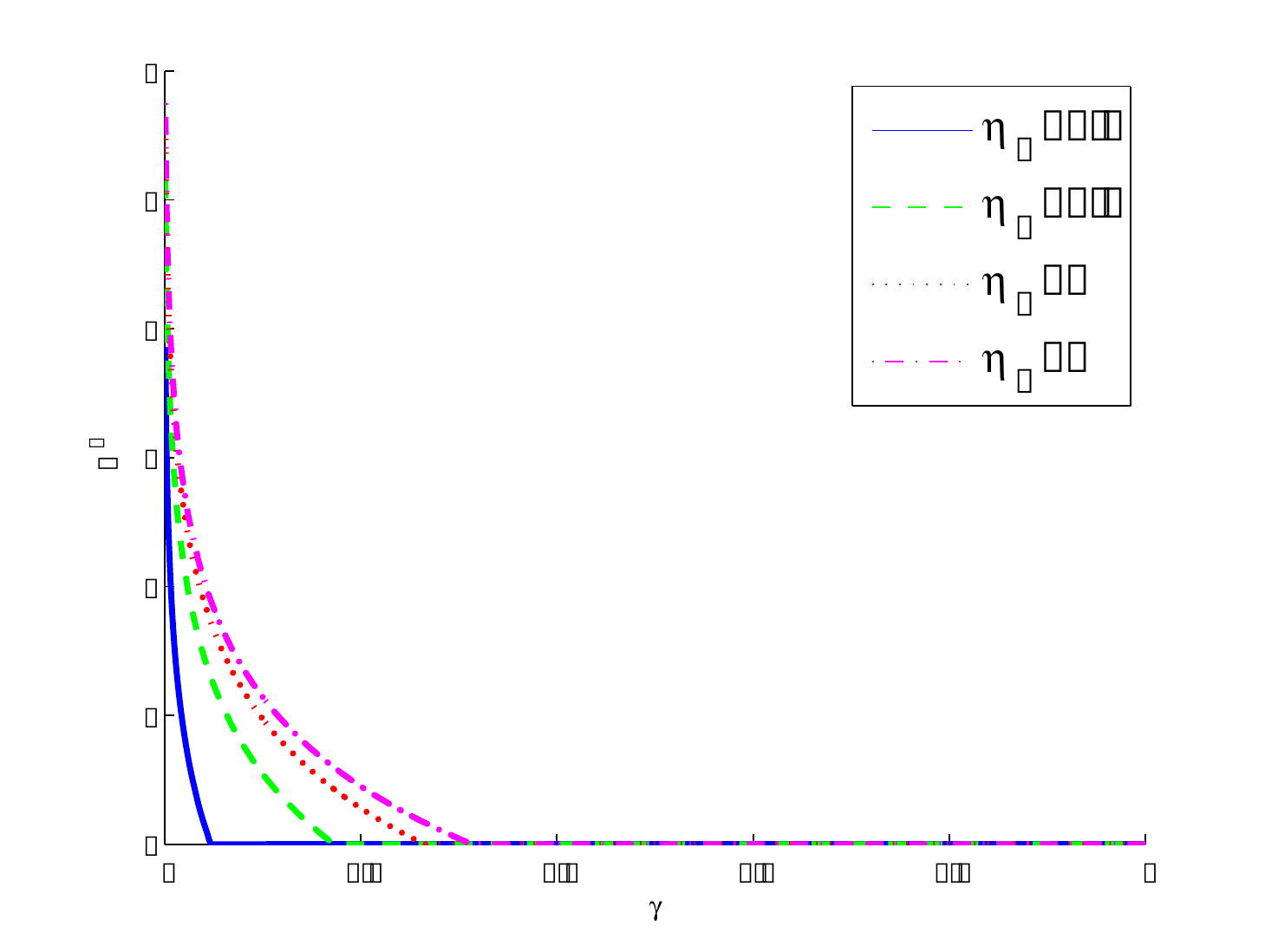} \\
(iii) $\eta_-$ & (iv) $\eta_+$
 \vspace{0.7cm}
\\
\includegraphics[scale=0.5]{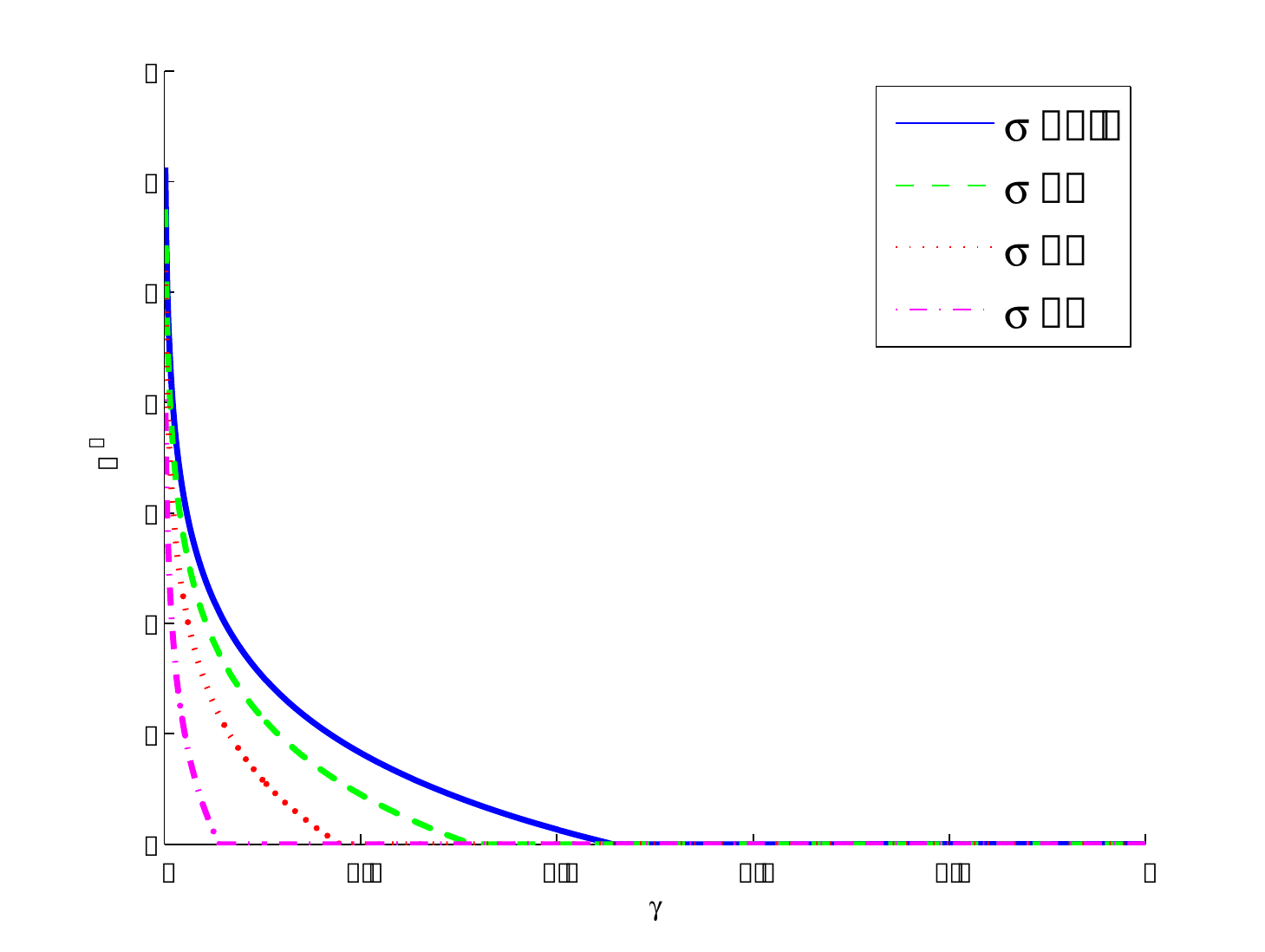}  & \includegraphics[scale=0.5]{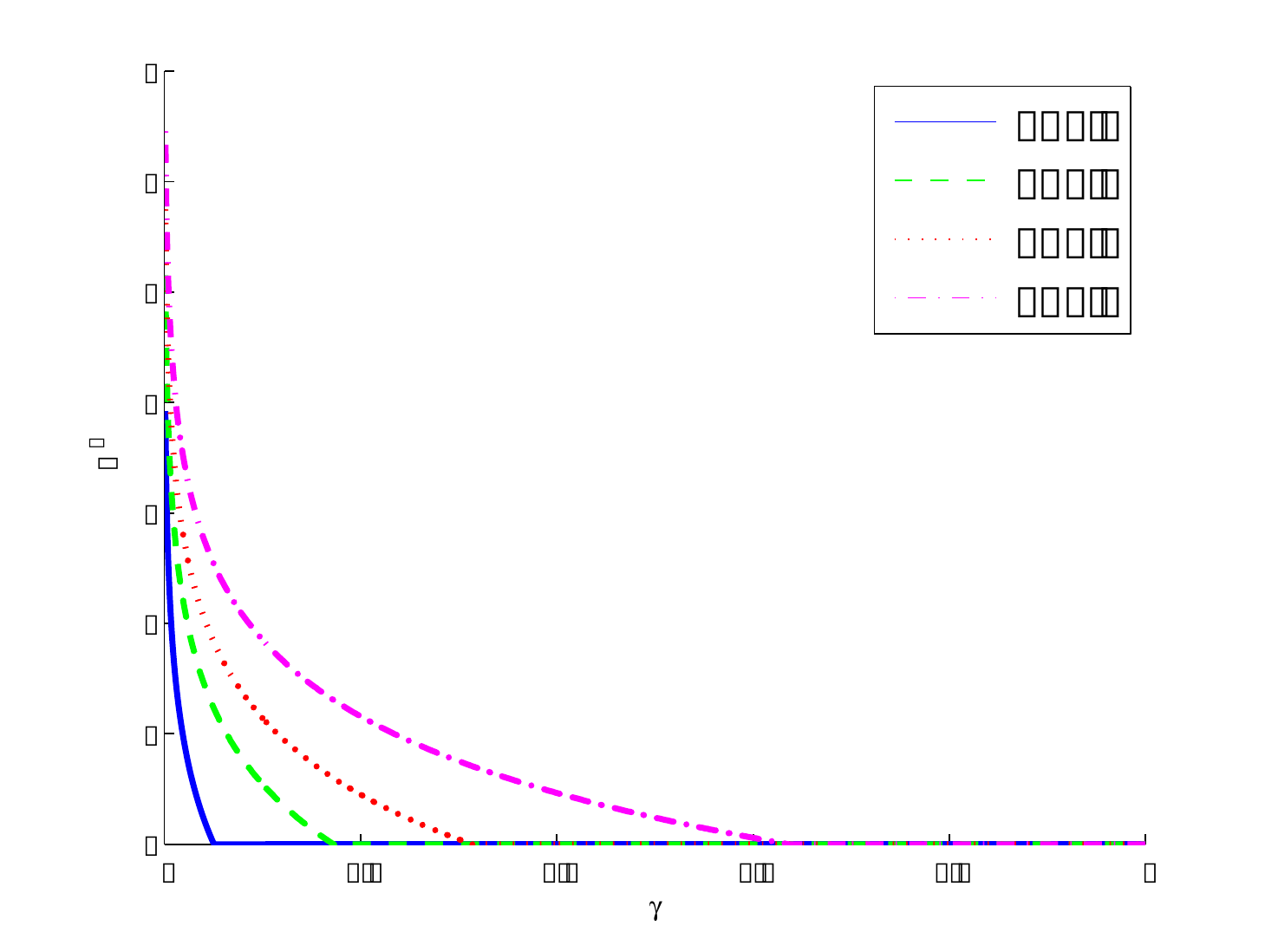}  \\
(v) $\sigma$ & (vi) $p$
\end{tabular}
\end{minipage}
\caption{The optimal threshold level $A^*$ with respect to various parameters: the parameters are
$\mu=-1$, $\sigma = 1$, $\eta_- = 1.0$, $\eta_+ = 2.0$, $p=0.5$,
$\lambda = 1$ unless they are specified.}
\label{fig:sensitivity_analysis}
\end{center}
\end{figure}

Figure \ref{fig:sensitivity_analysis} shows how the optimal threshold level
 changes with respect to each parameter when $q=0.05$. The results obtained in (i)-(iv) and (vi) are consistent with our intuition because these parameters determine the overall drift $\overline{u}$, and  $A^*$ is expected to decrease in $\overline{u}$.
  We show in (v) how it changes with respect to the diffusion coefficient $\sigma$; although it does not play a part in determining $\overline{u}$,
 we see that $A^*$ is in fact decreasing in $\sigma$.
 This is related to the fact that, as $\sigma$ increases, the probability of jumping over the level zero decreases.

\subsection{The spectrally negative \lev case with a general $h$}
We now consider the spectrally negative case and verify the results obtained in Section \ref{sec:spectrally_negative}. For the function $h$, we use the exponential utility function
\begin{align*}
h(x) = 1 - e^{-\rho x}, \quad x > 0.
\end{align*}
Here $\rho > 0$ is called the \emph{coefficient of absolute risk aversion}. It is well-known that $\rho = - h''(x)/h(x)$ for every $x > 0$, and, in particular, $h \equiv 1$ when
$\rho = \infty$.  We consider the \emph{tempered stable (CGMY) process} and the \emph{variance gamma process} with only downward jumps. The former has
a Laplace exponent
\begin{align*}
\psi(\beta) = c \beta + C \lambda^\alpha \Gamma (- \alpha) \left\{ \left( 1 + \frac \beta \lambda \right)^\alpha - 1 - \frac {\beta \alpha} \lambda \right\}
\end{align*}
and a \lev density given by
\begin{align*}
\Pi(\diff x) = C \frac {e^{-\lambda x}} {x^{1+\alpha}} 1_{\{x > 0\}} \diff x
\end{align*}
for some $\lambda > 0$ and $\alpha < 2$; see Surya \cite{Surya_2007} for the calculations. It has paths of bounded variation if and only if $\alpha < 1$.
 When $\alpha=0$, it reduces to the \emph{variance gamma} process; see Proposition 5.7.1 of Surya \cite{Surya_2007} for the form of its Laplace exponent.
We consider the case when $\sigma = 0$. The optimal threshold levels are computed by the bisection method using  \eqref{condition_spectrally_negative} with error bound $10^{-6}$.

Figure \ref{fig:results_spectrally_negative} shows the optimal threshold level $A^*$ as a function of $\gamma$ with various values of $\rho$.  We see that it is indeed monotonically decreasing in $\rho$. This can be also analytically verified  in view of the definition of $\Phi(A)$ because $h$ is decreasing in $\rho$ for every fixed $A$, and consequently the root $A^*$ must be decreasing in $\rho$. This is also clear because the regret function monotonically decreases in $\rho$.

\begin{figure}[htbp]
\begin{center}
\begin{minipage}{1.0\textwidth}
\centering
\begin{tabular}{cc}
\includegraphics[scale=0.5]{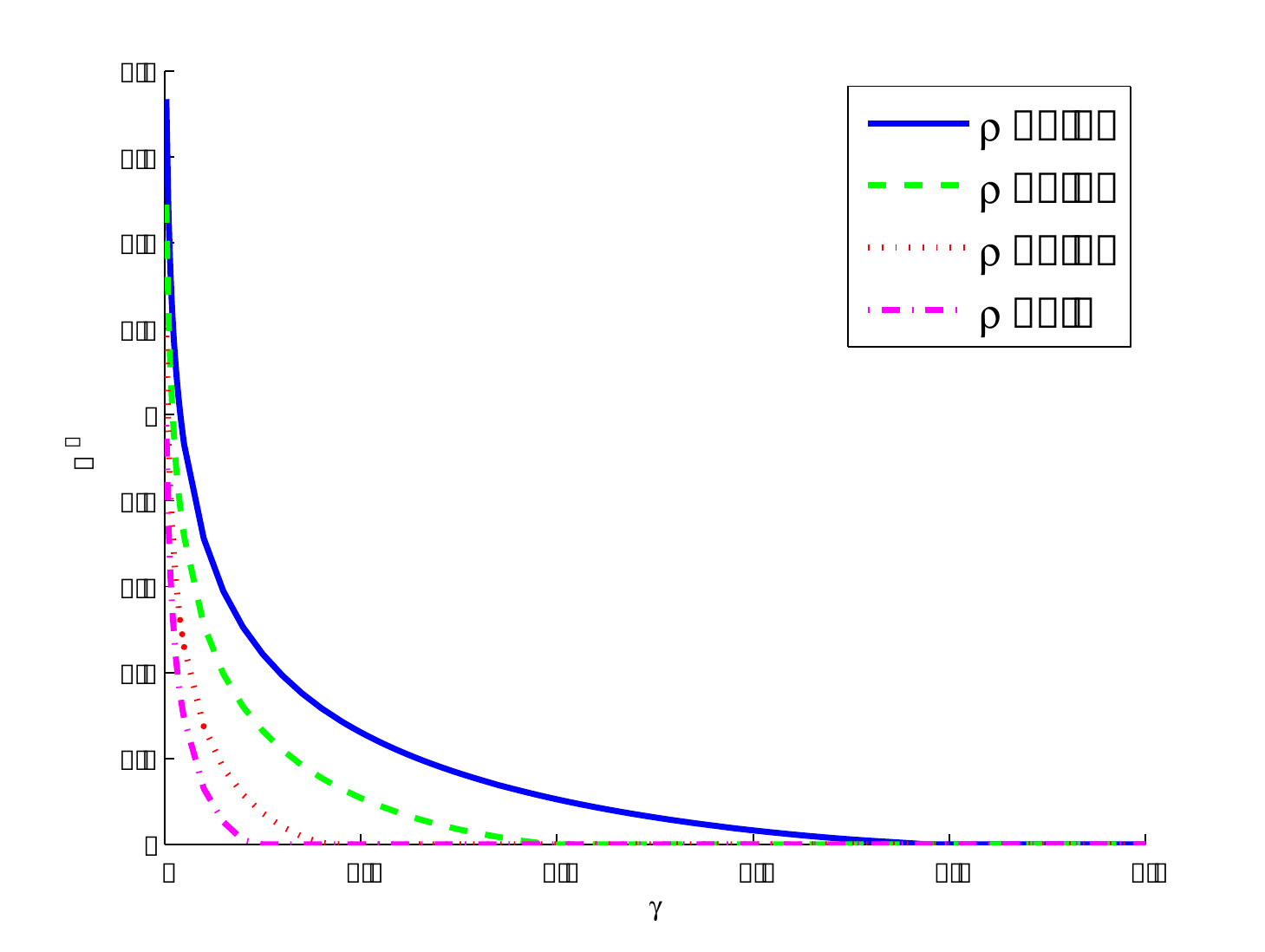}  & \includegraphics[scale=0.5]{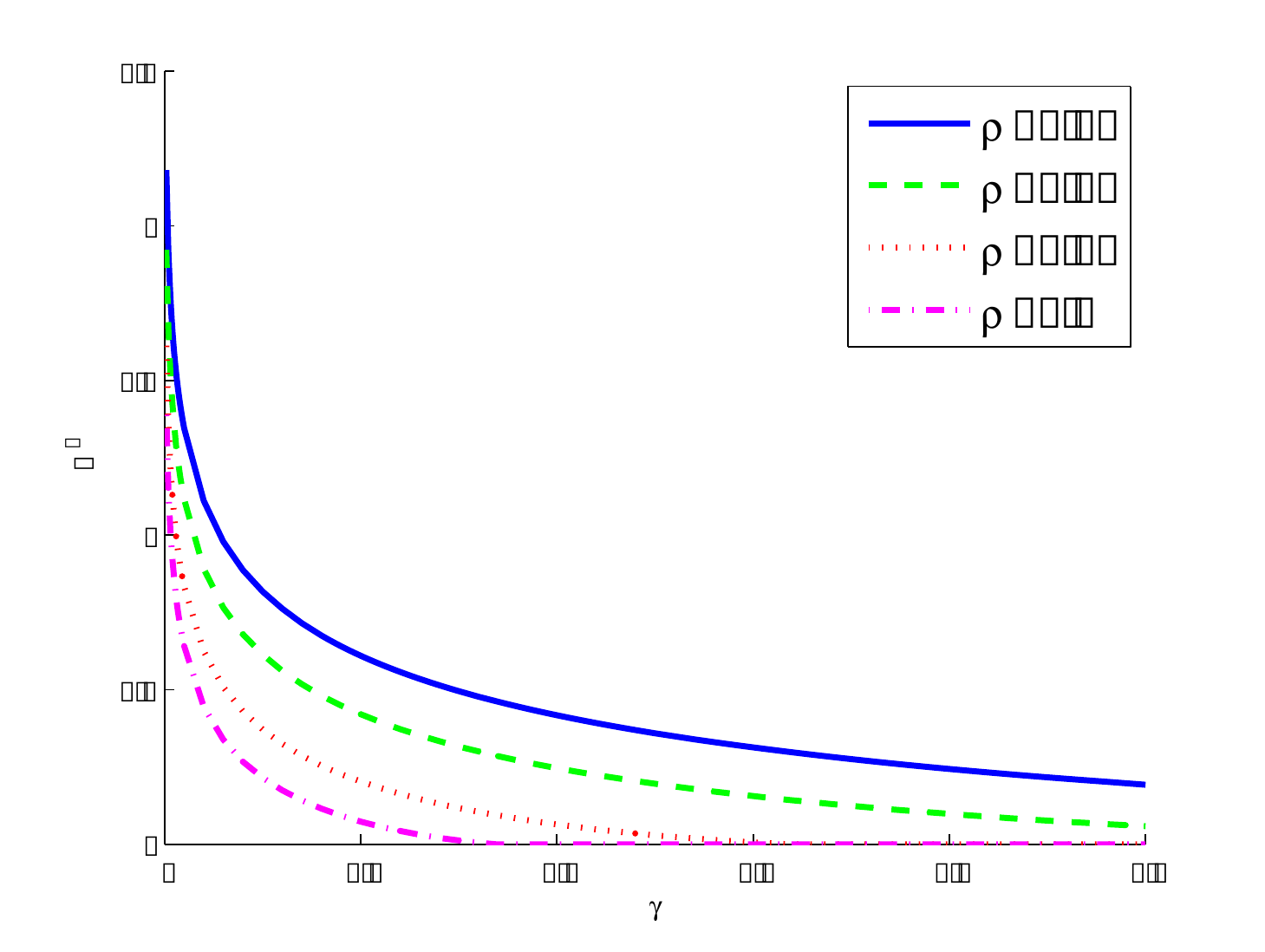} \\
(a) tempered stable (unbounded variation) & (b) tempered stable (bounded variation) 
\end{tabular}
\begin{tabular}{c}
\includegraphics[scale=0.5]{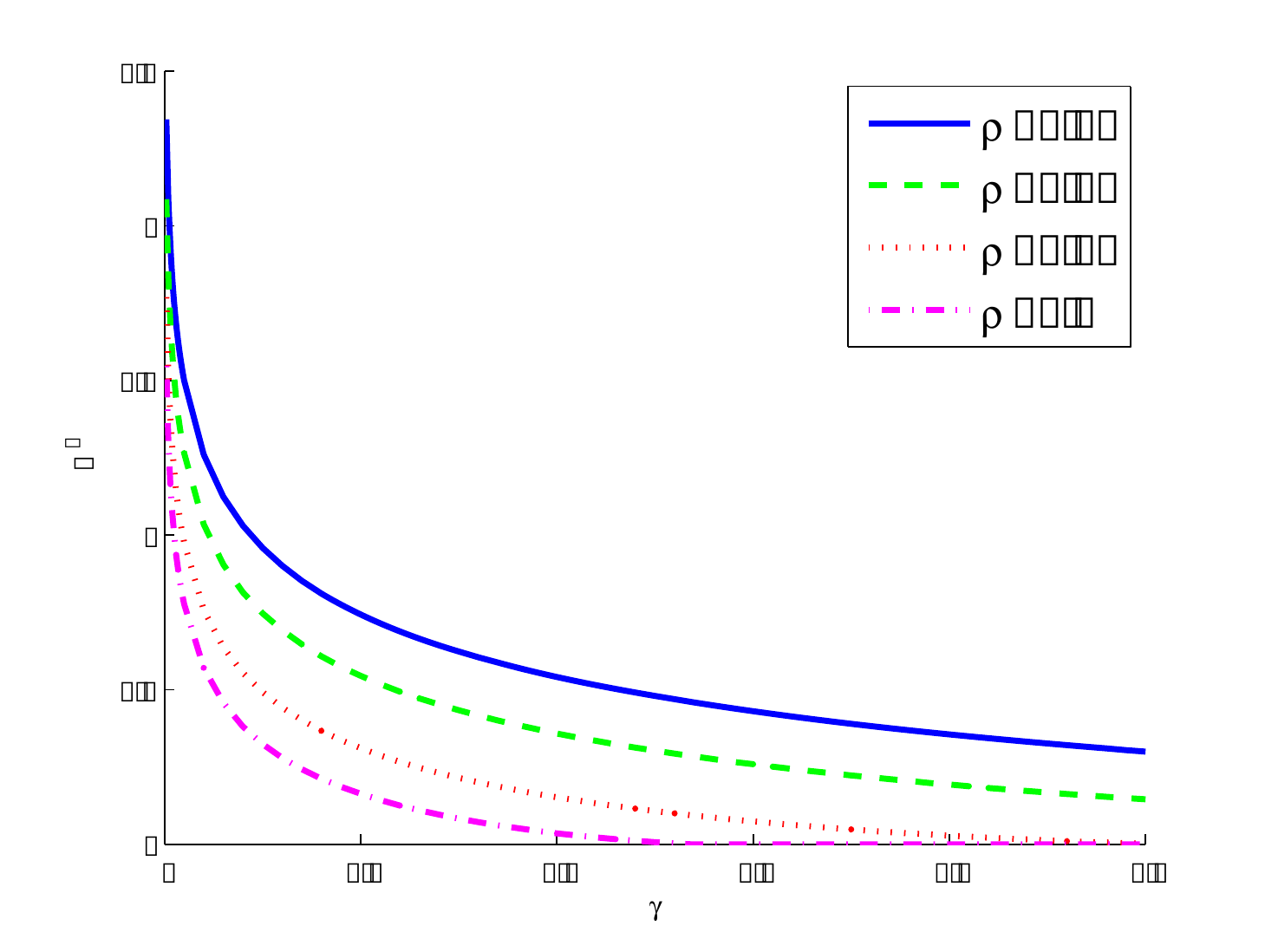}    \\
(c) variance gamma 
\end{tabular}
\end{minipage}
\caption{The optimal threshold level $A^*$ with various values of coefficients of absolute risk aversion $\rho$: (a) tempered stable (unbounded variation) with $\alpha = 1.5$, $\lambda = 2$, $C=0.05$ and $c = 0.05$, (b) tempered stable (bounded variation) with $\alpha = 0.8$, $\lambda = 2$, $C=0.075$ and $c = 0.05$,  and (c) variance gamma with with $\lambda = 2$, $C=0.075$ and $c = 0.05$.}
\label{fig:results_spectrally_negative}
\end{center}
\end{figure}

\subsection{Continuous and smooth fit}

We conclude this section by numerically verifying the continuous and smooth fit conditions.  Unlike the optimal threshold level $A^*$, the computation of the value function $\phi$ involves that of the scale function. Here we consider the spectrally negative \lev process with exponential jumps in the form \eqref{double_exponential_x} with $p=1$ and $\sigma \geq 0$.  We consider the bounded variation case ($\sigma=0$) and the unbounded variation case ($\sigma>0$).  We also set $h \equiv 1$. This is a special case of the spectrally negative \lev process with phase-type jumps, and its scale function can be obtained analytically as in Egami and Yamazaki \cite{Egami_Yamazaki_2010_2}.   In general, scale functions can be approximated by Laplace inversion algorithm by Surya \cite{Surya_2008} or the phase-type fitting approach by Egami and Yamazaki \cite{Egami_Yamazaki_2010_2}.  One drawback of the approximation methods of the scale function is that the error tends to explode as $x$ gets large (see \eqref{scale_function_asymptotic}).  Because our objective here is to accurately verify the continuous and smooth fit conditions, we use an example where an explicit form is known.  Notice that the threshold level $A^*$ can be computed independently of the scale function, and hence one can alternatively approximate the value function $\phi$ by simulation.   

Figure \ref{fig:smooth_continuous_fit} draws the stopping value $G$ as well as the value function $\phi$ for both the bounded and unbounded variation cases.   The value function $\phi$ is indeed continuous at $A^*$ for the bounded variation case and differentiable for the unbounded variation case.  It can be seen that $G$ indeed dominates $\phi$.  While $G$ is monotonically increasing on $(0,\infty)$, $\phi$ is decreasing for large $x$ and is expected to converge to zero as $x \rightarrow \infty$.

\begin{figure}[htbp]
\begin{center}
\begin{minipage}{1.0\textwidth}
\centering
\begin{tabular}{cc}
\includegraphics[scale=0.5]{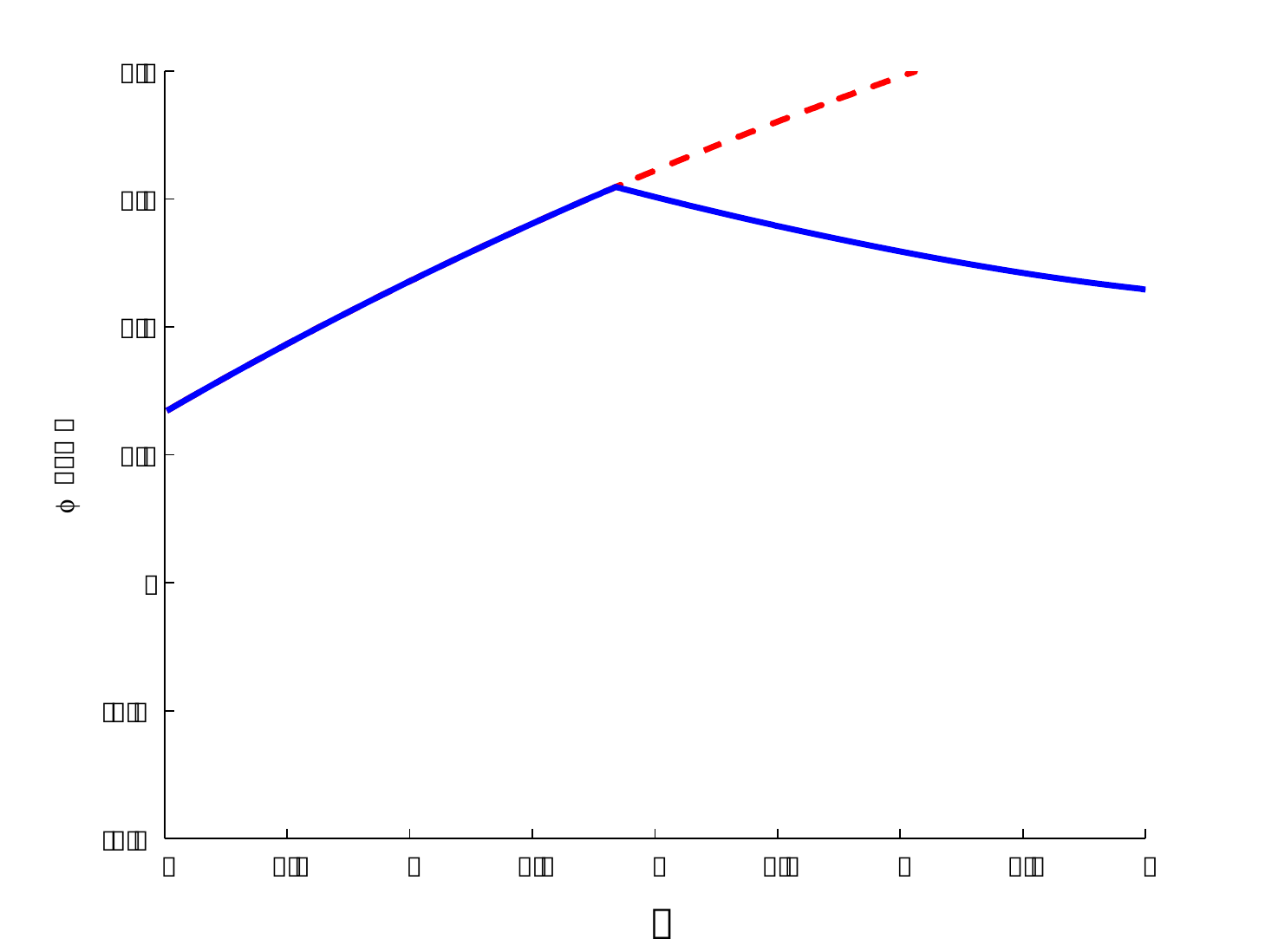}  & \includegraphics[scale=0.5]{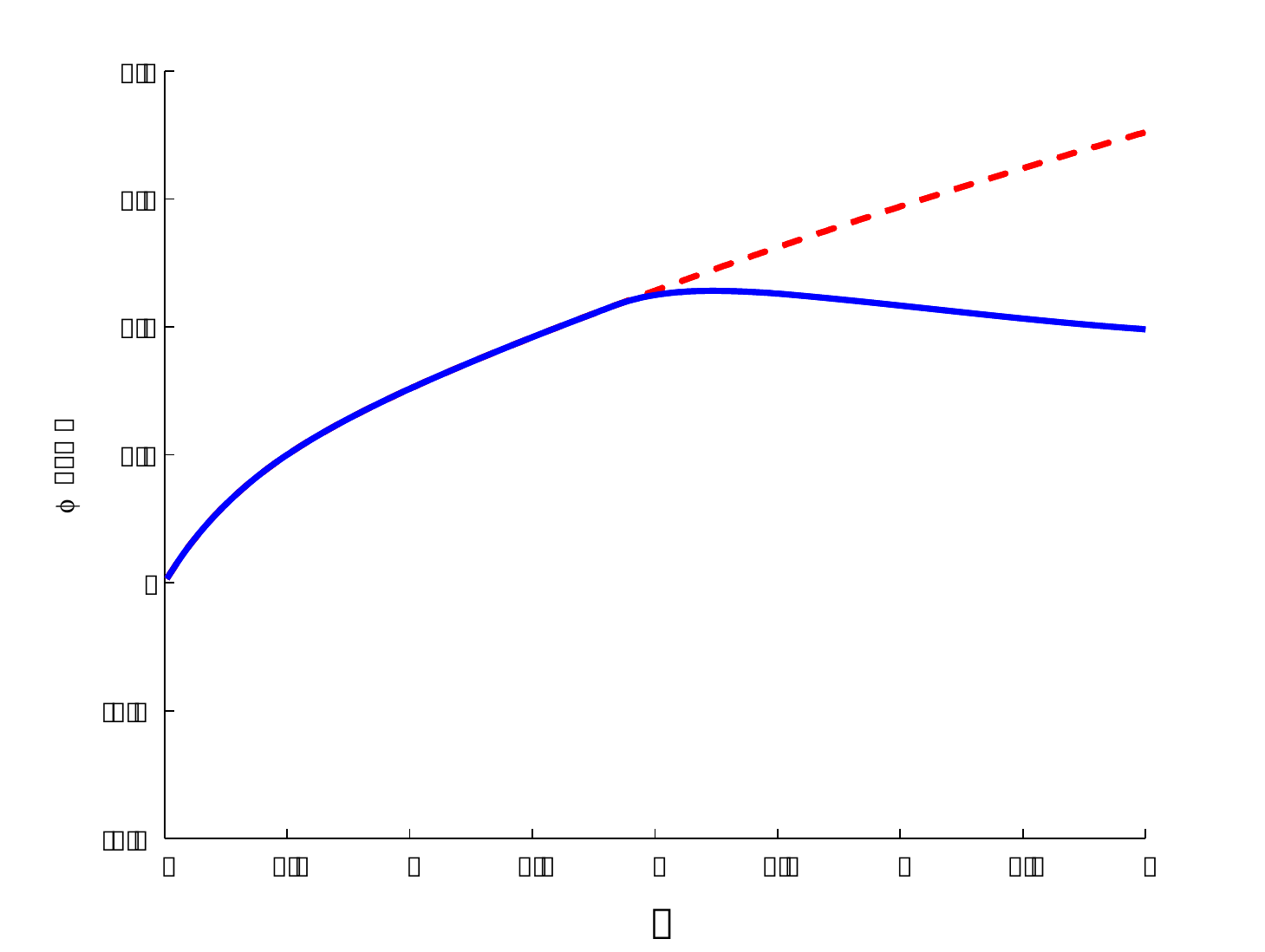} \\
(a) continuous fit & (b) smooth fit
\end{tabular}
\end{minipage}
\caption{Illustration of $\phi$ (solid) and $G$ (dotted) for (a) the bounded variation case with $\sigma = 0$ and $\mu = 0.3$ and (b) the unbounded variation case with $\sigma =0.5$ and $\mu = 0.175$.  Other parameters are $q = 0.05$, $h \equiv 1$, $\lambda = 0.5$, $\eta_- = 1$ and $\gamma = 0.04$.}
\label{fig:smooth_continuous_fit}
\end{center}
\end{figure}

\appendix

\section{Proofs}

\subsection{Proof of Lemma \ref{lemma_verification}}
We shall prove for the case $q > 0$ and then extend it to the case $q = 0$ and $\overline{u} < 0$.  
We first prove the following for the proof of Lemma \ref{lemma_verification}.

\begin{lemma} \label{lemma_inequality_generator}
Fix $q > 0$ and $x \in
\mathbb{R}$. Suppose that a function $w: \mathbb{R} \rightarrow \mathbb{R}$ in a neighborhood of $x > 0$ is given by
\begin{align*}
w(x) =  k + \sum_{i=1,2} k_i e^{-\xi_{i,q}x}
\end{align*}
for some $k$, $k_1$ and $k_2$ in $\R$. Then we have
\begin{align*}
q (w(x)-k) = \mathcal{L} w(x)- \lambda \left[ \int_{-\infty}^\infty  w(x+z)  f(z) \diff z  - \left( k + \sum_{i=1,2} k_i  \left( \frac {p \eta_-} {\eta_- - \xi_{i,q}} + \frac {(1-p)\eta_+} {\eta_+ + \xi_{i,q}} \right) e^{-\xi_{i,q}x} \right) \right].
\end{align*}
\end{lemma}
\begin{proof}
Because $\psi(-\xi_{1,q})  = \psi(-\xi_{2,q}) = q$, we have
\begin{align*}
q (w(x) - k) =  \sum_{i=1,2} k_i  \psi(-\xi_{i,q}) e^{-\xi_{i,q}x}.
\end{align*}
Moreover, the right-hand side equals by
(\ref{laplace_exponent_double_exponential})
\begin{align*}
 &\sum_{i=1,2}k_i  \left[ -\mu \xi_{i,q} + \frac 1 2 \sigma^2 (\xi_{i,q})^2 + \lambda \left( \frac {p \eta_-} {\eta_- - \xi_{i,q}} + \frac {(1-p)\eta_+} {\eta_+ + \xi_{i,q}} - 1\right) \right] e^{-\xi_{i,q}x} \\
&=- \lambda (w(x)-k) + \sum_{i=1,2}  k_i  \left[ - \mu \xi_{i,q} + \frac 1 2 \sigma^2 (\xi_{i,q})^2 + \lambda \left( \frac {p \eta_-} {\eta_- - \xi_{i,q}} + \frac {(1-p)\eta_+} {\eta_+ + \xi_{i,q}} \right) \right] e^{-\xi_{i,q}x} \\
&= \frac 1 2 \sigma^2 w''(x) + \mu w'(x) - \lambda w(x) + \lambda k  + \lambda \sum_{i=1,2} k_i  \left( \frac {p \eta_-} {\eta_- - \xi_{i,q}} + \frac {(1-p)\eta_+} {\eta_+ + \xi_{i,q}} \right) e^{-\xi_{i,q}x} \\
&= \mathcal{L} w(x)- \lambda \left[ \int_{-\infty}^\infty  w(x+z)  f(z) \diff z  - \left( k + \sum_{i=1,2} k_i  \left( \frac {p \eta_-} {\eta_- - \xi_{i,q}} + \frac {(1-p)\eta_+} {\eta_+ + \xi_{i,q}} \right) e^{-\xi_{i,q}x} \right) \right],
\end{align*}
as desired.
\end{proof}

\begin{proof}[Proof of Lemma \ref{lemma_verification} when $q> 0$]

(i) Suppose $A^* > 0$.
By  Lemma \ref{lemma_inequality_generator} above, we have $\mathcal{L} \phi(x)-q \phi(x)$ equals, for every $x > {A^*}$,
\begin{align}
\lambda \left[ \int_{-\infty}^\infty  \phi(x+z)  f(z) \diff z  -  \sum_{i=1,2} (L_{i,q} - C_{i,q})  \left( \frac {p \eta_-} {\eta_- - \xi_{i,q}} + \frac {(1-p)\eta_+} {\eta_+ + \xi_{i,q}} \right) e^{-\xi_{i,q}x}  \right], \label{difference_l_q_1}
\end{align}
and, for every $0 < x < {A^*}$,
\begin{align}
- (q + \lambda) \frac \gamma q  + \lambda \left[ \int_{-\infty}^\infty
\phi(x+z)  f(z) \diff z  +
\sum_{i=1,2} C_{i,q} \left( \frac {p \eta_-} {\eta_- - \xi_{i,q}} +
\frac {(1-p)\eta_+} {\eta_+ + \xi_{i,q}} \right) e^{-\xi_{i,q}x}
\right]  \label{difference_l_q_2}
\end{align}
by using (\ref{eq:G-expression}) and  (\ref{def:L}).

\underline{Proof of (\ref{equality_continuation})}.
We only need to show (\ref{difference_l_q_1}) equals zero.  Notice that the integral can be split into four parts, and, by using (\ref{about_C})-(\ref{about_C_2}), we have
\begin{align*}
\int_0^\infty \phi(x+z) f(z) \diff z &= \sum_{i = 1,2} (L_{i,q} - C_{i,q}) \frac {(1-p)\eta_+} {\eta_+ + \xi_{i,q}}   e^{-\xi_{i,q}x}, \\
\int_{-(x-{A^*})}^0 \phi(x+z) f(z) \diff z &= \sum_{i = 1,2} (L_{i,q} - C_{i,q}) \frac {p \eta_-} {\eta_- - \xi_{i,q}} e^{-\xi_{i,q}x} - \sum_{i = 1,2} (L_{i,q} - C_{i,q}) \frac {p \eta_-} {\eta_- - \xi_{i,q}}  e^{- \eta_- (x-{A^*}) - \xi_{i,q}{A^*} }, \\
\int^{-(x-{A^*})}_{-x} \phi(x+z) f(z) \diff z &= \frac {p \gamma} q e^{- \eta_- (x-{A^*})}  - \sum_{i = 1,2} C_{i,q} \frac {p \eta_-} {\eta_- - \xi_{i,q}} e^{- \eta_- (x-{A^*}) - \xi_{i,q} {A^*} },  \\
\int_{-\infty}^{-x}  \phi(x+z) f(z) \diff z &= p e^{-\eta_- x}.
\end{align*}
Putting altogether, (\ref{difference_l_q_1}) equals
\begin{multline*}
- \sum_{i = 1,2} L_{i,q} \frac {p \eta_-}
{\eta_- - \xi_{i,q}}  e^{- \eta_- (x-{A^*}) -  \xi_{i,q}{A^*}} + {
\frac {p\gamma} q} e^{- \eta_- (x-{A^*})} + p e^{-\eta_- x} \\ = p e^{-\eta_- x} \left[ 1 - {\frac \gamma q} \frac {\xi_{1,q} \xi_{2,q}} {(\eta_- - \xi_{1,q}) (\xi_{2,q} - \eta_-)} e^{\eta_- {A^*}} \right],
\end{multline*}
and this vanishes because of the way $A^*$ is chosen in \eqref{def_a_tilde_q}.

\underline{Proof of (\ref{inequality_stopping})}. We shall show that (\ref{difference_l_q_2}) is decreasing in $x$. Note that $\int_{-\infty}^\infty \phi(x+z)  f(z) \diff z$ can be split into four parts with
\begin{align*}
\int_{{A^*}-x}^\infty \phi(x+z)  f(z) \diff z &= (1-p) \sum_{i=1,2} (L_{i,q} - C_{i,q}) \frac {\eta_+} {\xi_{i,q} + \eta_+} e^{-\eta_+(A^*-x)  -  \xi_{i,q}{A^*}},  \\
\int_{0}^{{A^*}-x} \phi(x+z)  f(z) \diff z &= (1-p) \left[  \frac \gamma q \left( 1 - e^{-\eta_+ ({A^*} - x)} \right) - \sum_{i=1,2}  C_{i,q}  \frac {\eta_+} {\xi_{i,q} + \eta_+} \left(  e^{- x \xi_{i,q}} -  e^{-\eta_+(A^*-x)  -  \xi_{i,q}{A^*}}\right) \right], \\
\int_{-x}^0 \phi(x+z)  f(z) \diff z &= \frac {p \gamma} q - p \sum_{i=1,2}  C_{i,q} \frac {\eta_-} {\eta_- - \xi_{i,q}} e^{- \xi_{i,q} x}, \\
\int_{-\infty}^{-x} \phi(x+z)  f(z) \diff z &= p e^{-\eta_- x},
\end{align*}
and after some algebra, we have
\begin{multline*}
\int_{-\infty}^\infty \phi(x+z)  f(z) \diff z  + \sum_{i=1,2} C_{i,q} \left( \frac {p \eta_-} {\eta_- - \xi_{i,q}} + \frac {(1-p)\eta_+} {\eta_+ + \xi_{i,q}} \right) e^{- \xi_{i,q}x} \\
=  \frac \gamma q  + p e^{-\eta_- x} + (1-p) e^{\eta_+ (x- {A^*})} \left(-  \frac \gamma q +  \sum_{i=1,2} L_{i,q} \frac {\eta_+} {\xi_{i,q} + \eta_+} e^{-  \xi_{i,q}{A^*}}   \right),
\end{multline*}
which by (\ref{def:L}) equals to
\begin{align*}
 \frac \gamma q  + p e^{-\eta_- x} + (1-p) e^{-\eta_+ ({A^*}-x)} \frac \gamma {\xi_{2,q}-\xi_{1,q}} \frac { \xi_{1,q} \xi_{2,q}} q \left(\frac 1 {\xi_{2,q} + \eta_+} - \frac 1 {\xi_{1,q} + \eta_+} \right).
\end{align*}
Hence we see that  (\ref{difference_l_q_2}) or $\mathcal{L} \phi(x) - q \phi(x)$ equals
\begin{align*}
-\gamma + \lambda \left[  p e^{-\eta_- x} - (1-p) e^{- \eta_+ ({A^*}-x)} \frac { \xi_{1,q} \xi_{2,q}} q \frac \gamma {(\xi_{2,q} + \eta_+) (\xi_{1,q} + \eta_+)} \right]
\end{align*}
and is therefore decreasing in $x$ on $(0,A^*)$.  We now only need to show that $\lim_{x \uparrow A^*} \left( \mathcal{L} \phi(x) - q \phi(x) \right) > 0$.

For every $x > A^*$,
\begin{align*}
\delta''(x) = \frac \gamma q \frac {\xi_{1,q} \xi_{2,q}}{\xi_{2,q} - \xi_{1,q}}  \left[  \xi_{1,q} e^{-\xi_{1,q} (x - {A^*}) }  -  {\xi_{2,q}}  e^{-\xi_{2,q} (x - {A^*}) } \right],
\end{align*}
and hence, after taking $x \downarrow A^*$, \begin{align*}
\delta''({A^*+}) =  - \frac \gamma q \xi_{1,q} \xi_{2,q} < 0.
\end{align*}
Consequently, by the continuous and smooth fit conditions $\delta(A^*+) = \delta'(A^*+) = 0$ and the definition of $\mathcal{L}$ in (\ref{generator_double_exponential}), we have
\begin{align*}
 \lim_{x \uparrow {A^*}} \left( \mathcal{L} \phi(x) - q \phi(x) \right) > \lim_{x \downarrow {A^*}} \left( \mathcal{L} \phi(x) - q \phi(x) \right) = 0,
\end{align*}
as desired.

(ii) Suppose $A^* = 0$. By (\ref{v_negative_case}), we have
\begin{align*}
\phi(x)  = C \left( e^{-\xi_{1,q} x} - e^{-\xi_{2,q} x} \right), \quad x > 0
\end{align*}
where
\begin{align*}
C := \frac {(\eta_- - \xi_{1,q})(\xi_{2,q}-\eta_-)} {\eta_- (\xi_{2,q} - \xi_{1,q})}.
\end{align*}
By  Lemma \ref{lemma_inequality_generator} above, we have $\mathcal{L} \phi(x)-q \phi(x)$ equals, for every $x > 0$,
\begin{align*}
\lambda \left[ \int_{-\infty}^\infty  \phi(x+z)  f(z) \diff z  -  \sum_{i=1,2} C \left( \frac {p \eta_-} {\eta_- - \xi_{i,q}} + \frac {(1-p)\eta_+} {\eta_+ + \xi_{i,q}} \right) e^{-\xi_{i,q}x}  \right],
\end{align*}
and this vanishes after some algebra.
\end{proof}

\begin{proof}[Proof of Lemma \ref{lemma_verification} when $q= 0$]  We extend the results above to the case $q=0$ and $\overline{u} < 0$.  Solely in this proof, let us emphasize the dependence on $q$ and use $A^*_q$, $\phi_q(\cdot)$, $G_q(\cdot)$, $\delta_q(\cdot)$, and $\delta_{A,q}(\cdot)$ with a specified discount rate $q \geq 0$.  We shall show that $\mathcal{L} \phi_0 (x) = \lim_{q \rightarrow 0}\mathcal{L} \phi_q(x)$ for all $x > 0$.  First notice that $A^*_q \rightarrow A^*_0$ as $q \rightarrow
0$. 

Clearly, $G_q(x) \rightarrow G_0(x)$ by the monotone convergence theorem.  Furthermore, we have
\begin{align*}
G_q'(x) &= \frac \gamma {q \eta_-} \xi_{1,q} \xi_{2,q} \sum_{i = 1,2} l_{i,q} e^{-\xi_{i,q}x} \xrightarrow{q \downarrow 0} \frac \gamma {|\overline{u}|} \left[ 1 + \frac {\xi_{2,0} - \eta_-} {\eta_-} e^{- \xi_{2,0}x}\right] = G_0'(x), \\
G_q''(x) &= -\frac \gamma {q \eta_-} \xi_{1,q} \xi_{2,q} \sum_{i = 1,2}\xi_{i,q} l_{i,q} e^{-\xi_{i,q}x} \xrightarrow{q \downarrow 0} - \frac \gamma {|\overline{u}|} \left[\frac {(\xi_{2,0} - \eta_-) \xi_{2,0}} {\eta_-} e^{- \xi_{2,0}x}\right] = G_0''(x).
\end{align*}

Fix $x > A_0^*$.  We suppose $A_0^* = 0$ and focus on $q \in [0,q_0]$ with $q_0 > 0$ sufficiently small such that $A_q^* = 0$ for all $0 \leq q \leq q_0$. We have, by applying the monotone convergence theorem on \eqref{v_negative_case}, $\phi_q(x) \rightarrow \phi_0(x)$ as $q \rightarrow 0$.  Furthermore,
\begin{align*}
\phi'_q(x) =  \frac {(\eta_- - \xi_{1,q})(\xi_{2,q}-\eta_-)} {\eta_- (\xi_{2,q} - \xi_{1,q})} \left(-\xi_{1,q} e^{-\xi_{1,q} x} + \xi_{2,q} e^{-\xi_{2,q} x} \right) \xrightarrow{q \downarrow 0}  \xi_{2,0}e^{-\xi_{2,0} x} = \phi_0'(x), \\
\phi_q''(x) =  \frac {(\eta_- - \xi_{1,q})(\xi_{2,q}-\eta_-)} {\eta_- (\xi_{2,q} - \xi_{1,q})} \left(\xi_{1,q}^2 e^{-\xi_{1,q} x} - \xi_{2,q}^2 e^{-\xi_{2,q} x} \right) \xrightarrow{q \downarrow 0}  -\xi_{2,0}^2 e^{-\xi_{2,0} x} = \phi_0''(x).
\end{align*}
Suppose $A_0^* > 0$ and focus on $q \in [0,q_0]$ with $q_0 > 0$ sufficiently small such that $x > A_q^*$ for all $0 \leq q \leq q_0$.  Note that
\begin{align*}
|\delta_q(x) - \delta_0(x)| \leq |\delta_{q}(x) - \delta_{A_q^*,0}(x)| + |\delta_{A_q^*,0}(x) - \delta_{0}(x)|
\end{align*}
where on the right-hand side the former vanishes as $q \rightarrow 0$ by the monotone convergence theorem in view of \eqref{diff_v_g} and the latter vanishes because $A_q^* \rightarrow A_0^*$; hence $\delta_q(x) \rightarrow \delta_0(x)$ as $q \rightarrow 0$. Furthermore, 
\begin{align*}
\delta_q'(x) &= \frac \gamma q \frac {\xi_{1,q} \xi_{2,q}} {\xi_{2,q}- \xi_{1,q}} \left[ - e^{- \xi_{1,q}(x -A^*_q)} +  e^{- \xi_{2,q}(x-A^*_q)}\right] \xrightarrow{q \downarrow 0} -\frac \gamma {|\overline{u}|} \left[ 1 -  e^{- \xi_{2,0}(x-A^*_0)} \right] = \delta_0'(x), \\
\delta_q''(x) &= \frac \gamma q \frac {\xi_{1,q} \xi_{2,q}} {\xi_{2,q}- \xi_{1,q}} \left[ \xi_{1,q} e^{- \xi_{1,q}(x-A^*_q)} - \xi_{2,q}  e^{- \xi_{2,q}(x-A^*_q)}\right] \xrightarrow{q \downarrow 0} - \frac \gamma {|\overline{u}|} \xi_{2,0}  e^{- \xi_{2,0}(x-A^*_0)} = \delta_0''(x).
\end{align*}


In summary, we have   $\phi_q(x) \rightarrow \phi_0(x)$, $\phi_q'(x) \rightarrow \phi_0'(x)$ and $\phi_q''(x) \rightarrow \phi_0''(x)$ as $q \rightarrow 0$ for every $x > 0$.  Moreover, by Remark \ref{remark:phi_bounded}-(3), via the dominated convergence theorem,
\begin{align*}
\int_{-\infty}^\infty \phi_0(x+z) f(z) \diff z = \lim_{q \rightarrow \infty}\int_{-\infty}^\infty \phi_q(x+z) f(z) \diff z, \quad x \geq 0.
\end{align*}
Consequently, $\lim_{q \rightarrow 0} \left( \mathcal{L} \phi_q(x) - q \phi_q(x) \right) = \mathcal{L} \phi_0(x)$.  This together with the result for $q > 0$ shows the claim.
\end{proof}  

\subsection{Proof of Lemma \ref{lemma_measure}}
Because $h$ is continuous on $(0,\infty)$, it is Borel measurable.  Hence there exists a converging sequence of simple functions $(h^{(n)})_{n\in \N}$ increasing to $h$ in the form
\begin{align*}
h^{(n)} (y) := \sum_{i=1}^{l(n)} b_i^{(n)} 1_{\left\{ y \in
B_{n,i}\right\}}, \quad n \geq 1,
\end{align*}
for some $l: \mathbb{N} \rightarrow \mathbb{N}$, $\{ b_i^{(n)}; n \geq 1, 0 \leq i \leq l(n) \}$ and Borel measurable sets $\left\{B_{n,i}; n \geq 1, 0 \leq i \leq l(n)
\right\}$. See page 99 of Cinlar and Vanderbei \cite{Cinlar_Vanderbei}.

Then the right-hand side of (\ref{eq:M-integral}) is, by the monotone convergence theorem,
\begin{multline*}
\int_{\R} M^{(A,q)}(\omega,\diff y) h(y)  =  \int_{\R} M^{(A,q)}(\omega,\diff y) \lim_{n \rightarrow \infty} \sum_{i=1}^{l(n)} b_i^{(n)} 1_{\left\{ y \in B_{n,i}\right\}} \\
= \lim_{n \rightarrow \infty} \sum_{i=1}^{l(n)} b_i^{(n)} \int_{\R}
M^{(A,q)}(\omega,\diff y)  1_{\left\{ y \in B_{n,i}\right\}} = \lim_{n
\rightarrow \infty} \sum_{i=1}^{l(n)} b_i^{(n)} M^{(A,q)}(\omega,
B_{n,i}).
\end{multline*}
This is indeed equal to the left-hand side of (\ref{eq:M-integral}) because, by the monotone
convergence theorem,
\begin{multline*}
\int_0^{\tau_A(\omega)} e^{-qt} h(X_t(\omega))\diff t
= \int_0^{\tau_A(\omega)} e^{-qt} \lim_{n \rightarrow \infty} \sum_{i=1}^{l(n)} b_i^{(n)} 1_{\left\{ X_t(\omega) \in B_{n,i}\right\}}\diff t \\
= \lim_{n \rightarrow \infty}\sum_{i=1}^{l(n)} b_i^{(n)}
\int_0^{\tau_A(\omega)} e^{-qt}  1_{\left\{ X_t(\omega) \in
B_{n,i}\right\}}\diff t = \lim_{n \rightarrow \infty} \sum_{i=1}^{l(n)}
b_i^{(n)} M^{(A,q)}(\omega, B_{n,i}),
\end{multline*}
as desired.

\subsection{Proof of Lemma \ref{violation_risk_spectrally_negative}}
Let  $N(\cdot, \cdot)$ be the Poisson random measure for $-X$ and $\underline{X}_t := \min_{0 \leq s \leq t} X_s$ for all $t \geq 0$. By
the compensation formula  (see, e.g., Theorem 4.4 in Kyprianou
\cite{Kyprianou_2006}), we have
\begin{align*}
R^{(q)}_x (\tau_A)
&= \E^x \left[ \int_0^\infty  \int_0^\infty N(\diff t, \diff u) e^{ -q t} 1_{\{X_{t-} - u \leq 0, \,  \underline{X}_{t-} > A \}}\right] \\
&= \E^x \left[ \int_0^\infty  e^{ -q t} \diff t \int_0^\infty \Pi (\diff u) 1_{\{ X_{t-} - u \leq 0, \,  \underline{X}_{t-} > A \}}\right] \\ &= \int_0^\infty \Pi (\diff u) \int_0^\infty \diff t \left[ e^{ -q t} \p^x \{X_{t-} \leq u,  \underline{X}_{t-} > A \} \right] \\
&= \int_0^\infty \Pi (\diff u) \int_0^\infty \diff t \left[ e^{ -q t} \p^x \{X_{t-} \leq u, \tau_A \geq t \} \right].
\end{align*}
 By using the $q$-resolvent kernel that appeared in Lemma \ref{lem:3}, we have
for $u
> A$
\begin{align*}
\int_0^\infty \diff t \left[ e^{ -q t} \p^x \{X_{t-} \leq u, \tau_A \geq t \} \right] &= \int_A^{u} \diff y \left[ e^{-\zeta_q (y-A)} W^{(q)} (x-A) -  W^{(q)} (x-y)\right] \\
&= \int_0^{u-A} \diff z \left[ e^{-\zeta_q z} W^{(q)} (x-A) -  W^{(q)} (x-z-A)\right] \\
&=  \frac 1 {\zeta_q} W^{(q)}(x-A)  \left( 1 - e^{-\zeta_q(u-A)}\right) - \int_0^{u-A} \diff z W^{(q)}(x-z-A),
\end{align*}
and it is zero on $0 \leq u \leq A$.
Substituting this, we have
\begin{align*}
R^{(q)}_x (\tau_A) =
\int_A^\infty \Pi (\diff u)  \left\{ \frac 1 {\zeta_q} W^{(q)}(x-A) \left( 1 -  e^{-\zeta_q (u-A)} \right) - \int_0^{u-A}  \diff z  W^{(q)} (x-A-z) \right\}.
\end{align*}
By \eqref{eq:at-zero}, we have
\begin{align*}
R^{(q)}_x (\tau_A)
&=
\frac 1 {\zeta_q} W^{(q)}(x-A) \int_A^\infty \Pi (\diff u) \left( 1 -  e^{-\zeta_q (u-A)} \right) \\ &\qquad - \int_A^x \Pi (\diff u)  \int_0^{u-A}  \diff z   W^{(q)} (x-A-z)  - \int_x^\infty \Pi (\diff u)  \int_0^{x-A}  \diff z   W^{(q)} (x-A-z) \\
&=
\frac 1 {\zeta_q} W^{(q)}(x-A) \int_A^\infty \Pi (\diff u) \left( 1 -  e^{-\zeta_q (u-A)} \right)  \\ &\qquad - \frac 1 q \int_A^x \Pi (\diff u)   \left( Z^{(q)}(x-A) - Z^{(q)} (x-u) \right)   - \frac 1 q \int_x^\infty \Pi (\diff u)   \left( Z^{(q)}(x-A) - 1 \right) \\
&=
\frac 1 {\zeta_q} W^{(q)}(x-A) \int_A^\infty \Pi (\diff u) \left( 1 -  e^{-\zeta_q (u-A)} \right) - \frac 1 q \int_A^\infty \Pi (\diff u)  \left( Z^{(q)}(x-A) - Z^{(q)} (x-u) \right),
\end{align*}
as desired.

\subsection{Proof of Lemma \ref{lemma_difference_spectrally_negative}}
Fix $0 < A < x$.  We have
\begin{align*}
\frac \partial {\partial A} R_x^{(q)}(\tau_A) &= \left( W^{(q)}(x-A) - \frac 1 {\zeta_q} W^{(q)'} (x-A) \right)  \int_{A}^\infty \Pi(\diff u) \left( 1 - e^{-\zeta_q (u-A)}\right) \\ &= - \frac 1 {\zeta_q}  W_{\zeta_q}'(x-A) e^{\zeta_q (x-A)}  \int_{A}^\infty \Pi(\diff u) \left( 1 - e^{-\zeta_q (u-A)}\right),
\end{align*}
and because we can write
\begin{align*}
\E^{x}\left[ \int_0^{\tau_{A}} e^{-qt} h(X_t) \diff t \right] = e^{\zeta_q x} W_{\zeta_q} (x-A) \int_A^\infty e^{-\zeta_q y} h(y) \diff y - \int_A^x W^{(q)}(x-y) h(y) \diff y,
\end{align*}
we have
\begin{align*}
&\frac {\partial} {\partial A} \E^{x}\left[ \int_0^{\tau_{A}} e^{-qt} h(X_t) \diff t \right] \\ &= - W_{\zeta_q}'(x-A) e^{\zeta_qx} \int_A^\infty e^{-\zeta_qy}   h(y) \diff y - W_{\zeta_q}(x-A) e^{\zeta_q (x-A)} h(A) + W^{(q)} (x-A) h(A) \\
&= - W_{\zeta_q}'(x-A) e^{\zeta_q (x-A)} \int_0^\infty e^{-\zeta_qy}   h(y+A) \diff y.
\end{align*}
Summing up these, we obtain
\begin{align*}
\frac {\partial} {\partial A} \delta_A(x) = - W_{\zeta_q}'(x-A) e^{\zeta_q (x-A)} \Phi(A).
\end{align*}
Here, because $W_{\zeta_q}'(x-A)>0$ and $\Phi(A)$ is decreasing in $A$ and attains zero at $A^*$, we have
\begin{align}
\frac {\partial} {\partial A} \delta_A(x)  > 0 \Longleftrightarrow \Phi(A) < 0 \Longleftrightarrow A > A^*. \label{equivalence_derivative_a}
\end{align}

Now suppose $A^* > 0$, we have
\begin{align*}
\phi(x) = G(x) + \delta_{A^*} (x) \leq G(x) + \lim_{A \uparrow x} \delta_A(x) \leq G(x), \quad x > A^*
\end{align*}
where the last inequality holds because continuous fit holds
everywhere for the unbounded variation case and because $\lim_{A
\uparrow x}\delta_A(x) = W^{(q)}(0) \Phi(x) < 0$ by \eqref{equivalence_derivative_a} for the bounded
variation case (by noting that $x > A^*$). The case with $A^*=0$
holds in the same way by the definition that $\phi(x) =
\lim_{\varepsilon \downarrow 0} \phi_\varepsilon(x)$. Finally,
because $\phi=G$ on $(-\infty,A^*]$, the proof is complete.

\subsection{Proof of Lemma \ref{generator_smaller_than_zero_spectrally_negative}}
(1) When $x>A^*$, $\phi$ is defined in \eqref{phi_spectrally_negative_positive}.
Let $\widetilde{\phi}$ be defined such that $\widetilde{\phi}(x) =
\phi(x)$ for all $x > 0$ and $\widetilde{\phi}(x)=0$ for all $x \leq
0$.  We obtain
\begin{align*}
\mathcal{L} \widetilde{\phi}(x) - q \widetilde{\phi}(x) = \frac 1 q
\left( \int_x^\infty \Pi(\diff u) \mathcal{L} Z^{(q)} (x-u) - q
\int_x^\infty \Pi(\diff u) Z^{(q)} (x-u) \right)  =
-\Pi(x,\infty).
\end{align*}
Here the first equality holds by \eqref{scale_generator_zero} and because the operator $\mathcal{L}$ can go into the integrals thanks to the fact that $Z^{(q)}$ is $C^1$ everywhere and
$C^2$ on $\R \backslash \{0\}$ for the unbounded variation case and
it is $C^0$ everywhere and $C^1$ on $\R \backslash \{0\}$ for the
bounded variation case, and also to the fact that $x > A^*$.  The second equality holds by \eqref{eq:at-zero}. Because
$\widetilde{\phi}(x) = \phi(x)$ and $\mathcal{L}\phi (x) - \mathcal{L}
\widetilde{\phi}(x) = \Pi(x,\infty)$ for every $x > 0$, we have the claim.

(2) Suppose $0 < x < A^*$. We can write
\begin{align}
\phi(x) = \gamma W^{(q)}(x) \int_0^\infty e^{-\zeta_q y} h(y) \diff y  - \gamma   \int_0^x W^{(q)}(x-y) h(y) \diff y  + L(x), \quad x \in (-\infty, A^*) \backslash \{0\} \label{spec_negative_phi_stopping}
\end{align}
where $L(x) = 1_{\{x\leq 0 \}}$ for every $x \in \R$.  After applying $(\mathcal{L}-q)$, the first term vanishes thanks to \eqref{scale_generator_zero}.  For the second
term, by integration by parts, 
\begin{align*}
q \int_0^x W^{(q)}(x-y) h(y) \diff y &= \left[ h(y) (Z^{(q)}(x)-Z^{(q)}(x-y))\right]_{y=0}^{y=x}  - \int_0^x h'(y) (Z^{(q)}(x)-Z^{(q)}(x-y)) \diff y  \\
&=   h(x)  (Z^{(q)}(x)-1)   - \int_0^x h'(y) (Z^{(q)}(x)-Z^{(q)}(x-y)) \diff y  \\
&=   h(x)  (Z^{(q)}(x)-1)   - (h(x)-h(0))  {Z^{(q)}(x)} + \int_0^x h'(y)  {Z^{(q)}(x-y)} \diff y  \\
&=   - {h(x)}    +h(0)  {Z^{(q)}(x)} + \int_0^x h'(y) {Z^{(q)}(x-y)} \diff y  \\
&=  -  {h(M)}   + h(0) {Z^{(q)}(x)} +  \int_0^{M} h'(y)  {Z^{(q)}(x-y)} \diff y,
\end{align*}
for any $M > x$ where the last equality holds because $Z^{(q)}(x) = 1$ on
$(-\infty,0]$. 
The operator $(\mathcal{L}-q)$ can again go into the integral thanks to the smoothness of $Z^{(q)}$ as discussed in (1) and we obtain 
\begin{multline*}
(\mathcal{L}-q) \left( \int_0^x W^{(q)}(x-y) h(y) \diff y \right) = \frac 1 q (\mathcal{L}-q) \left( - {h(M)} +  \int_0^{M} h'(y) Z^{(q)}(x-y) \diff y \right)\\  = h(M) + \frac 1 q \int_0^{M} h'(y) (\mathcal{L}-q)  Z^{(q)}(x-y) \diff y =  h(M) - \int_x^{M} h'(y) \diff y =  h(x),
\end{multline*}
where the second to last equality holds by \eqref{scale_generator_zero}.
For the last term of \eqref{spec_negative_phi_stopping}, we have
\begin{align*}
(\mathcal{L}-q) L(x) = \int_0^\infty \left( L(x-u) - L(x) \right) \Pi(\diff u)=  \Pi(x,\infty).
\end{align*}

Putting altogether, we have noting that $x < A^*$ and $h$ is increasing,
\begin{align*}
(\mathcal{L}-q) \phi(x) &= \Pi(x,\infty) - \gamma h(x) \\ &\geq \int_{A^*}^\infty \Pi(\diff u) \left( 1 - e^{-\zeta_q (u-A^*)}\right) - \gamma h(x) \\
&= \int_{A^*}^\infty \Pi(\diff u) \left( 1 - e^{-\zeta_q (u-A^*)}\right) - \gamma \zeta_q \int_0^\infty e^{-\zeta_q y} h(x) \diff y \\ &\geq \int_{A^*}^\infty \Pi(\diff u) \left( 1 - e^{-\zeta_q (u-A^*)}\right) - \gamma \zeta_q \int_0^\infty e^{-\zeta_q y} h(A^*+y) \diff y,
\end{align*}
which is zero because $A^*$ satisfies \eqref{condition_spectrally_negative}.  This completes the proof.

\subsection{Proof of Proposition \ref{proposition_optimality_spectrally_negative}} \label{proof_proposition_optimality_spectrally_negative}
Due to the discontinuity of the value function at zero, we need to proceed carefully.  By \eqref{assumption_G} and Lemma \ref{lemma_difference_spectrally_negative}, we must have $1 = \phi(0) > \phi(0+)$.

We first construct a sequence of functions $\phi_n(\cdot)$ such that
(1) it is $C^2$ (resp.\ $C^1$) everywhere except at $A^*$ when $\sigma > 0$ ($\sigma = 0$), (2) $\phi_n(x) = \phi(x)$ on $x \in
(0,\infty)$ and (3) $\phi_n(x) \uparrow \phi(x)$ pointwise for every
fixed $x \in (-\infty,0)$ (with $\lim_{n \rightarrow \infty}\phi_n(0) = \phi(0+) < \phi(0)$).  It
can be shown along the same line as Remark \ref{remark:phi_bounded}-(3)
that $\phi(\cdot)$ is uniformly bounded.  Hence, we can choose so
that $\phi_n$ is also uniformly bounded for every $n \geq 1$. 

Because
$\phi'(x)=\phi'_n(x)$ and $\phi''(x)=\phi''_n(x)$ on $(0,\infty)
\backslash \{A^*\}$, we have $(\mathcal{L}-q) (\phi_n - \phi)(x) \leq 0$
for every fixed $x \in
(0,\infty) \backslash \{A^*\}$. 
Furthermore, by Lemma \ref{generator_smaller_than_zero_spectrally_negative}
\begin{align}
\E^x \left[ \int_0^{\tau} e^{-qs} ((\mathcal{L}-q) \phi_n(X_{s-}))
\diff s\right] \geq \E^x \left[ \int_0^{\tau} e^{-qs}
((\mathcal{L}-q) (\phi_n-\phi) (X_{s-})) \diff s\right]  > -\infty,
\quad \tau \in \S. \label{upper_bound_generator_v_n}
\end{align}
Here, the last lower bound is obtained because
\begin{align*}
\E^x \left[ \int_0^{\tau} e^{-qs} ((\mathcal{L}-q)
(\phi_n-\phi) (X_{s-})) \diff s\right]  \geq -K  \E^x \left[
\int_0^{\theta} e^{-qs} \Pi(X_{s-},\infty) \diff s\right]
\end{align*}
where $K < \infty$ is the maximum difference between $\phi$ and
$\phi_n$. Using $N$ as the Poisson random measure for $-X$ and $\underline{X}$ as the running infimum of $X$ as in the proof of Lemma \ref{violation_risk_spectrally_negative}, we have
by the compensation formula
\begin{align*}
\E^x \left[ \int_0^{\theta} e^{-qs} \Pi(X_{s-},\infty) \diff s\right] &=
\E^x \left[ \int_0^\infty \int_0^\infty e^{-qs} 1_{\{\theta \geq s,
\, u > X_{s-}\}} \Pi(\diff u) \diff s \right] \\ 
&= \E^x \left[ \int_0^\infty \int_0^\infty e^{-qs} 1_{\{\underline{X}_{s-} > 0,
\, u > X_{s-}\}} \Pi(\diff u) \diff s \right] \\ &= \E^x \left[
\int_0^\infty \int_0^\infty e^{-qs} 1_{\{\underline{X}_{s-} > 0, \, u >
X_{s-}\}} N (\diff s, \diff u) \right] \\ &= \E^x \left[ e^{-q
\theta} 1_{\{ X_\theta < 0, \, \theta < \infty \}}\right] < 1.
\end{align*}
By \eqref{upper_bound_generator_v_n}, we have uniformly in $n$
\begin{align}
\begin{split}
&\E^x \left[ \int_0^{\tau} e^{-qs}
|(\mathcal{L}-q) (\phi_n-\phi) (X_{s-})| \diff s\right] < \infty, \\
&\int_0^{\tau} e^{-qs}
|(\mathcal{L}-q) (\phi_n-\phi) (X_{s-})| \diff s < \infty, \quad \p^x-a.s.
\end{split} \label{dominate_two}
\end{align}
We remark here for the proof of Proposition \ref{proposition_optimality_case_1} that, in the double exponential case with $q=0$, there also exists a finite bound because the \lev measure is a finite measure and $\E^x \theta < \infty$ by assumption.

Notice that, although $\phi_n$ is not $C^2$ (resp.\ $C^1$) at $A^*$ for the case $\sigma > 0$ (the case of bounded variation), the Lebesgue measure of the
set where $\phi_n$ at which $X=A^*$ is zero and hence
$\phi_n''(A^*)$ ($\phi_n'(A^*)$) can be chosen arbitrarily; see also Theorem 2.1 of \cite{sulem}. By applying Ito's formula to
$\left\{ e^{-q {(t \wedge \theta)}} \phi_n(X_{t \wedge \theta }); t
\geq 0 \right\}$, we see that
\begin{align}
\left\{ e^{-q {(t \wedge \theta)}} \phi_n(X_{t \wedge \theta }) - \int_0^{t \wedge \theta } e^{-qs}\left( (\mathcal{L} - q) \phi_n (X_{s-})  \right) \diff s; \quad t \geq 0 \right\} \label{local_martingale}
\end{align}
is a local martingale.  Suppose $\left\{\sigma_k; k \geq 1 \right\}$ is the corresponding localizing sequence, namely,
\begin{align*}
\E^x \left[ e^{-q {(t \wedge \theta \wedge \sigma_k)}} \phi_n(X_{t \wedge \theta \wedge \sigma_k}) \right] = \phi_n(x) + \E^x \left[  \int_0^{t \wedge \theta \wedge \sigma_k} e^{-qs}\left( (\mathcal{L} - q) \phi_n (X_{s-})  \right) \diff s \right], \quad k \geq 1.
\end{align*}
Now by applying the dominated convergence theorem on the left-hand
side and Fatou's lemma on the right-hand side via
\eqref{upper_bound_generator_v_n}, we obtain
\begin{align*}
\E^x \left[ e^{-q {(t \wedge \theta)}} \phi_n(X_{t \wedge \theta}) \right] \geq \phi_n(x) + \E^x \left[  \int_0^{t \wedge \theta } e^{-qs}\left( (\mathcal{L} - q) \phi_n (X_{s-})  \right) \diff s \right].
\end{align*}
Hence  \eqref{local_martingale} is in fact a submartingale.

Now fix $\tau \in \S$. By the optional sampling theorem, 
for any $M \geq 0$, 
\begin{multline*}
\E^x \left[ e^{-q {(\tau \wedge M)}} \phi_n(X_{\tau \wedge M}) \right] \geq \phi_n(x) + \E^x \left[  \int_0^{\tau \wedge M} e^{-qs}\left( (\mathcal{L} - q) \phi_n (X_{s-})  \right) \diff s \right] \\
= \phi_n(x) + \E^x \left[  \int_0^{\tau \wedge M} e^{-qs}\left( (\mathcal{L} - q) \phi (X_{s-})  \right) \diff s \right] + \E^x \left[  \int_0^{\tau \wedge M} e^{-qs}\left( (\mathcal{L} - q) (\phi_n-\phi) (X_{s-})  \right) \diff s \right]
\end{multline*}
where the last equality holds because the expectation can be split
by \eqref{upper_bound_generator_v_n}. Applying the dominated
convergence theorem on the left-hand side and the monotone
convergence theorem on the right-hand side (here the integrands in
the two expectations are positive and negative, respectively) along
with Lemma \ref{generator_smaller_than_zero_spectrally_negative}, we
obtain
\begin{align}
\E^x \left[ e^{-q \tau} \phi_n(X_{\tau}) 1_{\{\tau < \infty \}}\right] \geq \phi_n(x) + \E^x \left[ \int_0^{\tau} e^{-qs} ((\mathcal{L}-q) (\phi_n -\phi)(X_{s-})) \diff s\right]. \label{supermtg_proof}
\end{align}
For the left-hand side, the dominated convergence theorem implies 
\begin{multline}
\lim_{n \rightarrow \infty} \E^x \left[ e^{-q \tau} \phi_n(X_\tau) 1_{\{\tau < \infty \}} \right] =  \E^x \left[ e^{-q \tau} \lim_{n \rightarrow \infty} \phi_n(X_\tau) 1_{\{\tau < \infty \}} \right]  \\ = \E^x \left[ e^{-q \tau} (\phi(X_{\tau }) 1_{\{X_\tau \neq 0\}} + \phi(0+) 1_{\{X_\tau = 0\}}) 1_{\{\tau < \infty \}} \right] \leq \E^x \left[ e^{-q \tau} \phi(X_{\tau }) 1_{\{\tau < \infty \}} \right],
\end{multline}
 which holds by equality when $\sigma = 0$ because $X$ creeps the level zero only when $\sigma > 0$.  For the right-hand side, by \eqref{dominate_two}, 
\begin{align}
\lim_{n \rightarrow \infty} \E^x \left[ \int_0^{\tau} e^{-qs} ((\mathcal{L}-q) (\phi_n -\phi)(X_{s-})) \diff s\right] =  \E^x \left[ \int_0^{\tau} e^{-qs} \lim_{n \rightarrow \infty} ((\mathcal{L}-q) (\phi_n -\phi)(X_{s-})) \diff s\right]. \label{expectation_limit_second}
\end{align}
Here, for every $\p^x$-a.e.\ $\omega \in \Omega$, because $X_{s-}(\omega) > 0$ for Lebesque-a.e.\ $s$ on $(0,\tau(\omega))$
\begin{align*}
\lim_{n \rightarrow \infty} ((\mathcal{L}-q) (\phi_n -\phi)(X_{s-}) (\omega)) = \int_{X_{s-}(\omega)}^\infty \Pi(\diff u) \lim_{n \rightarrow \infty} (\phi_n -\phi)(X_{s-}(\omega)-u) = 0.
\end{align*}
Hence \eqref{expectation_limit_second} vanishes.

Therefore, by taking $n \rightarrow \infty$ on both sides of \eqref{supermtg_proof} (note $\phi(x) = \phi_n(x)$ for any $x > 0$), we have
\begin{align*}
\phi(x)\leq \E^x \left[ e^{-q \tau} \phi(X_{\tau}) 1_{\{\tau < \infty \}} \right] \leq  \E^x \left[ e^{-q \tau}G(X_{\tau}) 1_{\{\tau < \infty \}} \right], \quad \tau \in \S,
\end{align*}
where the last inequality follows from Lemma \ref{lemma_difference_spectrally_negative}.
Finally, the stopping time $\tau_{A^*}$ attains the value function $\phi$ when $A^* > 0$ while $\tau_\varepsilon$ and $\phi_\varepsilon$ approximate $\phi$ by taking $\varepsilon$ sufficiently small when $A^*=0$.  This completes the proof.

\bibliographystyle{abbrv}
\bibliographystyle{apalike}

\bibliographystyle{agsm}
\small{\bibliography{alarmbib}}

\end{document}